\documentclass[11pt,a4paper,final]{article}
\usepackage[utf8]{inputenc}
\usepackage{amsmath}
\usepackage{amsfonts}
\usepackage{amssymb}
\usepackage{graphicx}
\usepackage{xcolor}
\usepackage[left=1in,right=1in,top=1in,bottom=1in]{geometry}

\newcommand{\qed}{\hfill$\rule{2mm}{3mm}$}
\newenvironment{proof}{\par{\noindent \bf Proof }}{\qed \par}

\usepackage{natbib}
\usepackage{amssymb}
\newcommand{\xhdr}[1]{\paragraph{ \textbf{#1.}}}

\newcommand{\dd}[1]{\frac{\partial}{\partial #1}}

\usepackage{amsmath}
\usepackage{natbib}
\usepackage{amsfonts}
\usepackage{amssymb}
\usepackage{graphicx}
\usepackage{subcaption}
\usepackage{multirow}
\usepackage{wrapfig}
\usepackage{url}
\usepackage{color}

\newtheorem{theorem}{Theorem}
\newtheorem{proposition}{Proposition}
\newtheorem{definition}{Definition}
\newtheorem{lemma}{Lemma}

\newtheorem{property}{Property}

\newcommand{\reals}{\mathbb{R}}

\newcommand{\cD}{\mathcal{D}}

\newcommand{\hl}[1]{\textcolor{black}{#1}}

\begin{document}
\title{Allocating Opportunities in a Dynamic Model of Intergenerational Mobility}

\author{
   \makebox[.45\textwidth]{Hoda Heidari}\\
   Carnegie Mellon University\\
   \url{hheidari@cmu.edu} \\
   \and
   \makebox[.45\textwidth]{Jon Kleinberg}\\
   Cornell University\\
   \url{kleinberg@cornell.edu} \\
}

\date{}

\maketitle

\begin{abstract}
Opportunities such as higher education can promote
{\em intergenerational mobility}, leading individuals to achieve
levels of socioeconomic status above that of their parents.
We develop a dynamic model for allocating such opportunities in 
a society that exhibits bottlenecks in mobility;
the problem of optimal allocation reflects a trade-off between
the benefits conferred by the opportunities in the current generation 
and the potential to elevate the socioeconomic status of 
recipients, shaping the composition of future generations in ways
that can benefit further from the opportunities.
We show how 
optimal allocations in our model arise as solutions to continuous
optimization problems over multiple generations, and we find in
general that these optimal solutions can favor recipients of
low socioeconomic status over slightly higher-performing individuals
of high socioeconomic status --- a form of socioeconomic affirmative
action that the society in our model discovers in the pursuit of
purely payoff-maximizing goals.
We characterize how the structure of the model can lead
to either temporary or persistent affirmative action,
and we consider extensions of the model with more complex processes
modulating the movement between different levels of
socioeconomic status.
\end{abstract}

\section{Introduction}

Intergenerational mobility --- 
the extent to which an individual's
socioeconomic status differs from the status of their
prior generations of family members ---
has emerged as a central notion in our understanding of inequality.
A large amount of empirical work has gone into estimating 
the extent of mobility for different subsets of society;
while many of the effects are complex and challenging to measure,
two broad and fairly robust principles emerge from this work. 
First, socioeconomic status is persistent across
generations: an individual's socioeconomic status is strongly
dependent on parental status.
As \citet{lee-solon-mobility-survey} write in the opening to their survey of this topic,
\textit{``Over the past two decades, a large body of research has documented
that the intergenerational transmission of economic status in the United States is much stronger than earlier sociological and economic
analyses had suggested''}. 
Second, certain types of opportunities can serve as
strong catalysts for socioeconomic mobility;
a canonical example is higher education, which has the potential to raise
an individual's socioeconomic status (and, by the previous principle,
that of their current or future children as well).
As \citet{chetty-mobility-trends} write,
\textit{``The fact that the college attendance is a good proxy for income
mobility is intuitive given the strong association between higher
education and subsequent earnings''}.

An important question from a social planning perspective is thus the choice of policy for allocating opportunities to people of different levels of socioeconomic status.
(Again, we can think of access to higher education as a
running example in this discussion.)
Many goals can motivate the choice of policy, including
the reduction of socioeconomic inequality and the prioritization of opportunities to those most in need.
Such goals are often viewed as operating in tension with 
the aim of maximizing the achievable payoff from the available opportunities,
which would seem to suggest targeting the opportunities
based only on the anticipated performance of the recipient, 
not their socioeconomic status.
In this view, society is implicitly being asked to choose between these goals;
this consideration forms a central ingredient in the 
informal discourse and debate around the allocation of opportunity.
But through all of this, a challenging question remains: to what
extent is the tension between these goals genuine, and to what extent
can they be viewed as at least partially in alignment?

\hl{
A large body of work in economics compares various allocation policies in terms of the above seemingly-competing criteria --- typically in simplified settings in which only two generations are considered. The literature includes seminal work by Nobel Laureate Garry Becker with Nigel Tomes~\citep{becker1986human} and by Glenn Loury~\citep{loury1981intergenerational}. 
In multigenerational settings, however, deriving the optimal policy becomes exceedingly challenging, and it has been highlighted as a class of open questions in this literature. For example, in his work on models of college admissions and intergenerational mobility, \citet{durlauf2008affirmative} notes:
\textit{``
A college admissions rule has intergenerational effects because it not only influences the human capital of the next generation of adults, but also affects the initial human capital of the generation after next. [...] Efficiency in student allocation [in this case] is far more complicated than before. I am unaware of any simple way of describing efficiency conditions for college assignment rules analogous to [the above setting].''}
In this work, we address this challenge and the associated open questions
concerning the behavior of multigenerational models.
A key ingredient in our progress on these questions is the development
of methods for working with a class of Markov Decision Processes that operate over
continuous states and continuous actions.
%
Our analysis of multigenerational models enables us to 
investigate the apparent tension between efficiency 
and fairness considerations in allocating opportunities.
}

\xhdr{Allocating Opportunities in a Payoff-Maximizing Society}
We work with a simple mathematical model
representing a purely payoff-maximizing society, operating over
multiple generations. \hl{As we discuss briefly in Section~\ref{sec:related-brief} and at more length in Appendix~\ref{sec:related}, our model is grounded in the types of models proposed
in economic theory work on these problems.}
The society must decide how to allocate opportunities in each generation across a population heterogeneous in its 
socioeconomic status.
The payoff to the society is the total performance of everyone who receives the opportunities,
summed (with discounting) over all generations.
Although the set-up of the model is highly streamlined, the analysis of
the model becomes quite subtle since society
must solve a continuous-valued dynamic programming problem over multiple generations.

What we find from the model is that the optimal solution will in general
tend to offer opportunities to individuals of lower socioeconomic status
over comparable individuals of higher socioeconomic status, even when
these competing individuals are predicted to have a slightly better performance from receiving the opportunity.
This is not arising because the optimal solution has any a priori
interest in reducing socioeconomic inequality (although such goals are important in their own right~\citep{forde2004taking}); rather it is strictly
trying to maximize payoff over multiple generations.
But given two individuals of equal predicted performance, the
one with lower socioeconomic status confers an added benefit to the
payoff function: their success would grow the size of the socioeconomically advantaged class, resulting in higher payoffs in future generations.
Because the difference in payoff contributions between these two individuals is strictly positive, the same decision would be optimal even if the individual of lower socioeconomic status had a slightly lower predicted performance from receiving the opportunity. The optimal solution should still favor the candidate with lower status in this case.

In other words, the society in this model discovers a form of
socioeconomic affirmative action in allocating opportunities,
based purely on payoff-maximizing motives.
The model thus offers a view of a system in which
reducing inequality is compatible with direct payoff maximization.
In this sense, our results belong to a genre of analyses
(popularized by Page \citep{page-difference-book} and others)
asserting that policies and interventions that we think of as motivated
by equity concerns, can also be motivated by purely performance-maximizing considerations: {\em even if} society only cares about performance, not equity, it should still (at least in the underlying models) 
undertake these policies.
\hl{In addition to providing a purely utilitarian motivation for socioeconomic affirmative action, our model provides novel insights regarding the shape and extent of effective  affirmative action policies by specifying the way in which criteria for receiving the opportunity should be adjusted based on socioeconomic status to maximize society's performance across multiple generations.}

We now give a rough overview of the model and results; 
a complete description of the model is provided in the following section.

\xhdr{A Model for Allocating Opportunities}
We consider a population that is partitioned into two groups
of different socioeconomic status:
$D$ (disadvantaged), consisting of a $\phi_0$ fraction of
the population, and $A$ (advantaged), 
consisting of a $\phi_1 = 1 - \phi_0$ fraction of the population.
Each agent $i$ (from either group)
has an ability $a_i$ drawn uniformly at random from
the interval $[0,1]$,

Society has the ability to offer an opportunity to an $\alpha$ fraction
of the population. \hl{Note that the parameter $\alpha$ specifies the inherent limitation on the amount of opportunities available. Since opportunities are limited, the society has to wrestle with the question of how to allocate them.}
An individual $i$ in group $D$ who is offered the opportunity has
a probability $\sigma a_i$ of succeeding at it, for a parameter 
$0 < \sigma < 1$.
An individual $i$ in group $A$ who is offered the opportunity has
a probability $\sigma a_i + \tau$ of succeeding at it, for the same
$\sigma$ and an additional parameter $0 < \tau \leq 1 - \sigma$ 
reflecting the advantage.
We will refer to the above quantities as the {\em success probabilities} of the agents. Success probabilities reflect various levels of performance when agents are offered the opportunity.

Anyone in group $D$ who is offered the opportunity and succeeds at it
moves up to group $A$.
Each individual is then replaced by one offspring of the same
socioeconomic status and the process continues to the next generation.
In the general form of the model, there is also
some probability that an individual's offspring does not perfectly
inherit their socioeconomic status.
The payoff to society is the number of individuals who succeed at
the opportunity summed over all generations, with the generation $t$
steps into future multiplied by $\gamma^t$ for a discount factor 
$0 < \gamma < 1$.

\xhdr{Summary of Results}
In any given generation, society's policy will consist of a threshold for group $D$ and a (possibly different) threshold for group $A$:
the opportunity is given to every individual whose success probability
is above the threshold for their group.
The optimal policy is given by a dynamic program over the continuous
set of all possible choices for the population composition $(\phi_0, \phi_1)$ as state variables.
We solve the most basic version of the model analytically. 
We computationally solve more complex versions of the model by discretizing the state space, then applying standard dynamic programming solutions for finite decision processes.

If the problem of allocating the opportunity only spanned a single
generation, then the payoff-maximizing policy would use the same threshold for both groups.
But given the discounting sum over multiple generations, we find that
society's optimal policy can, in general, use a lower threshold for group $D$ than for group $A$.
The difference in thresholds is a form of socioeconomic affirmative action,
and it arises due to the intuition discussed above:
boosting the number of individuals from group $D$ who receive the 
opportunity will increase the number of available candidates from 
group $A$ in future generations, each of whom provides a (discounted) payoff in future generations via their enhanced performance.
Finding the correct trade-off in allocating opportunity
thus involves a delicate balance 
between immediate and future utility.

\begin{figure*}[t!]
    \centering
    \begin{subfigure}[b]{0.3\textwidth}
        \includegraphics[width=\textwidth]{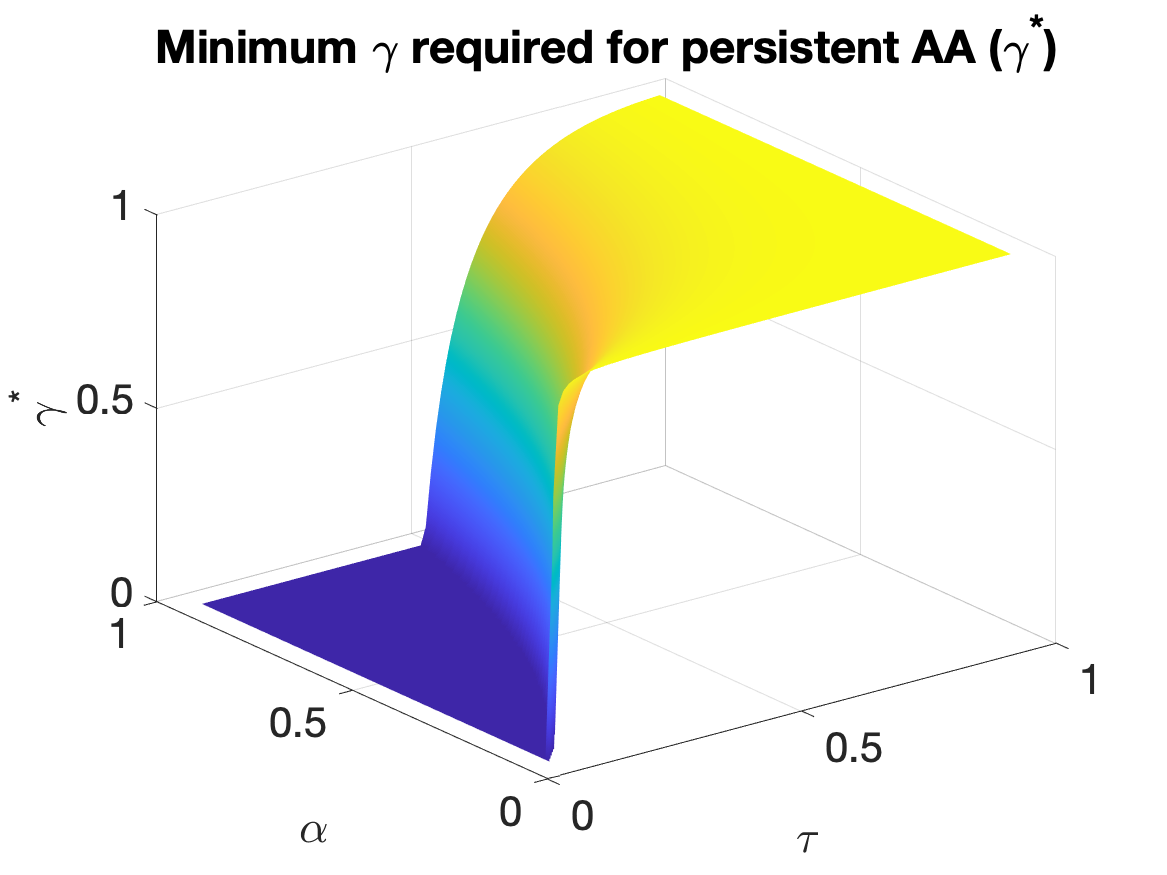}
        \caption{ }
        \label{fig:gamma_star}
    \end{subfigure}
    \begin{subfigure}[b]{0.3\textwidth}
        \includegraphics[width=\textwidth]{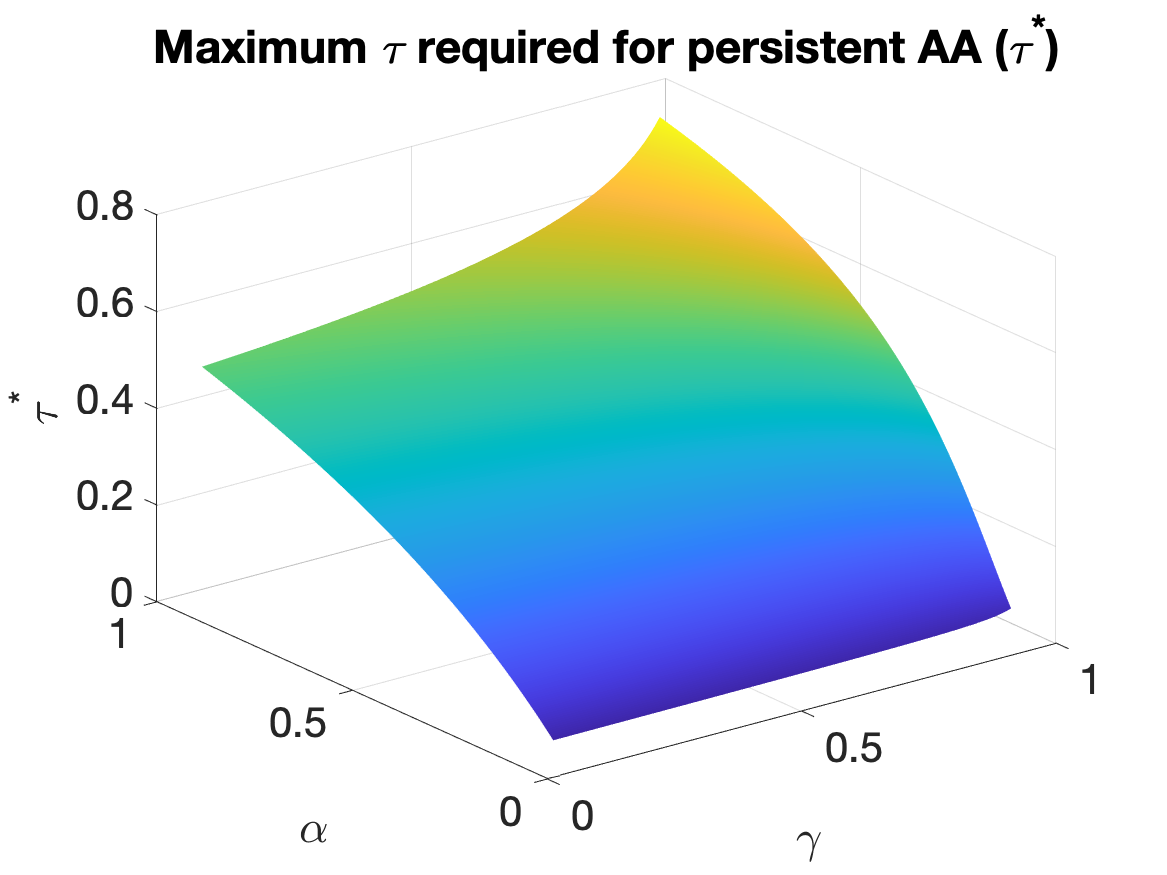}
        \caption{ }
        \label{fig:tau_star}
    \end{subfigure}
    \begin{subfigure}[b]{0.3\textwidth}
        \includegraphics[width=\textwidth]{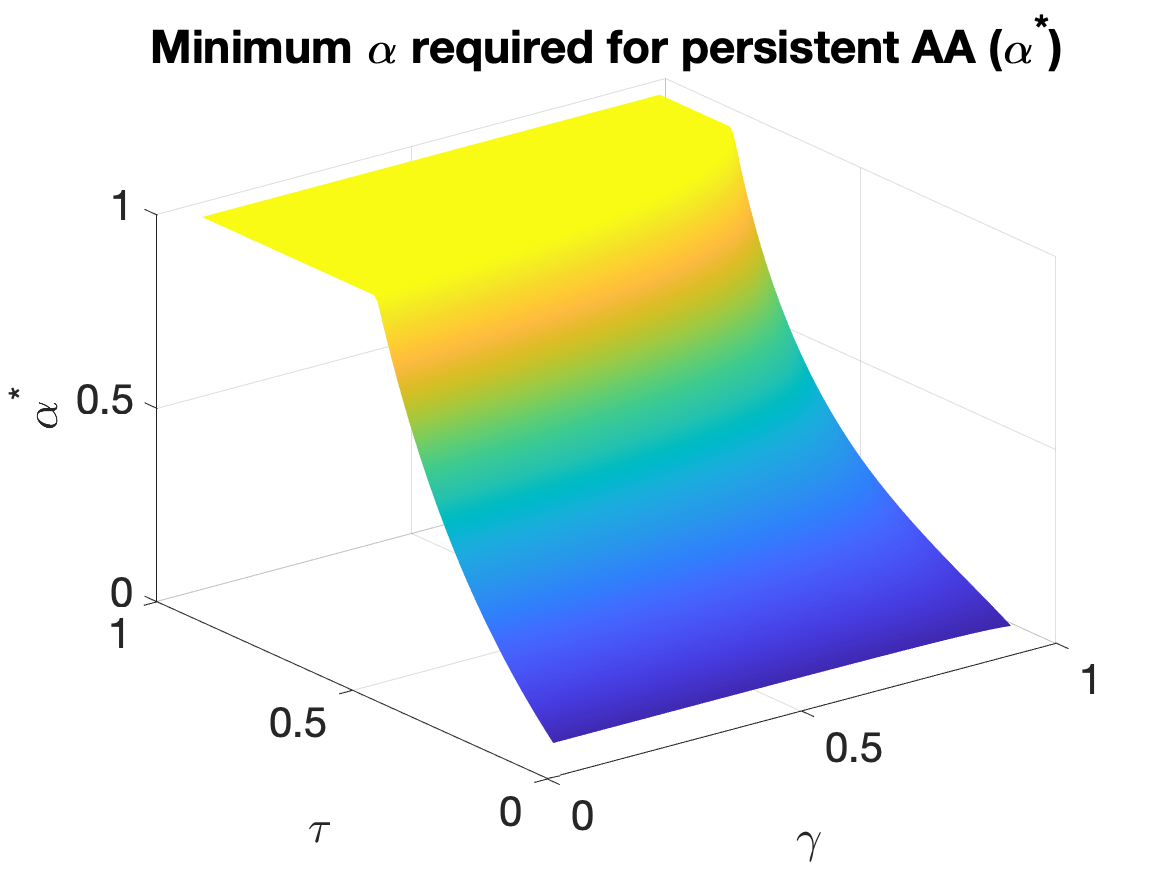}
        \caption{ }
        \label{fig:alpha_star}
    \end{subfigure}
   \caption{
Visualizing the triples $(\alpha, \tau, \gamma)$ for which the
optimal policy uses persistent affirmative actions.
(We set $\sigma = 1 - \tau$ for these plots.)
Points that lie above the surfaces in panels (a) and (c), and below
the surface in panel (b), correspond to parameter values yielding
persistent affirmative action.
(\ref{fig:gamma_star}) When $1-\frac{\alpha\sigma}{\tau}<0$, any $\gamma>0$ suffices for persistent affirmative action and eventually moving the entire population to group $A$. When $1-\frac{\alpha\sigma}{\tau}>0$, there exists some $\gamma < 1$ (and hence a finite level of patience $\left(\frac{\gamma}{1-\gamma}\right))$ that suffices for persistent affirmative action. 
   (\ref{fig:tau_star}) When $\tau$ is sufficiently large, the optimal policy does not use persistent affirmative action; this is because for a large $\tau$, the extent of affirmative action required to pick up the best performing members of $D$ is large---which in turn significantly reduces the immediate payoff. For any given value of $\alpha$, there exists a sufficiently small $\tau$ that guarantees persistent affirmative action.
   (\ref{fig:alpha_star}) When $\alpha$ is small relative to $\tau$, the optimal policy does not use persistent affirmative action; this is because the cost of picking the best performing members of $D$ is very high and a small $A$ group suffices for filling the available opportunities. Note that for some values of $\tau$, no matter how large $\alpha$ is, the optimal policy never employs persistent affirmative action.
   }\label{fig:star}
\end{figure*}

Whether socioeconomic affirmative action is employed by the optimal 
solution --- and the extent to which it is employed --- depends 
on the fraction $\phi_0$ of individuals from group $D$;
in the most basic model, the amount of affirmative action decreases monotonically as $\phi_0$ is reduced.
The extent of affirmative action is also determined
by the amount of opportunity available ($\alpha$), the dependence
of success on ability and socioeconomic status ($\sigma$ and $\tau$),
and society's patience in trading off immediate payoff in return for payoff from future generations ($\gamma$).
We characterize the optimal solution in this respect as a function
of these parameters, finding that for some regions of the parameter
space, the society employs {\em temporary affirmative action},
reducing the size of group $D$ to a given level before equalizing
thresholds in subsequent generations;
in other parts of the parameter space, the society employs
{\em persistent affirmative action}, in which the threshold for
group $D$ is strictly lower in every generation and the 
size of group $D$ converges to 0 over time.

Figure \ref{fig:star} provides some ways of describing the 
regions of parameter space in which the optimal solution
uses persistent affirmative action.
As the partitions of the space there make apparent, the interactions
among the key parameters is fairly subtle.
First, persistent affirmative action is promoted by large values of $\alpha$
and small values of $\tau$, since these make it easier to include 
high-performing members of group $D$ without a large difference in thresholds; 
and it is promoted by larger values of $\gamma$,
indicating greater concern for the payoffs in future generations.
One might have suspected that persistent affirmative action would
only be realized in the optimal solution in the limit as society's
patience (essentially $\gamma / (1 - \gamma)$) goes to infinity;
but in fact, a sufficiently large finite amount of patience is sufficient for the optimal policy to use persistent affirmative action.

In our model, we include a probabilistic background process by which
individuals can also move between groups $A$ and $D$; this reflects the idea that there are many mechanisms operating simultaneously
for socioeconomic mobility, and we are studying only one of these mechanisms via the opportunity under consideration.
The most basic version posits a single probability $p$ that each
individual independently loses their group membership and re-samples
it from the current distribution of group sizes.
We also consider a version of the model in which this probability of loss
of group membership is different for groups $A$ and $D$;
in this case, we are only able to solve the model computationally,
and these computational results reveal interesting non-monotonicities
in the amount of affirmative action employed
as a function of the relative size of group $D$ ($\phi_0$).

\hl{
\xhdr{Utilitarianism, Prioritarianism, and the Desert Principle} Our simple mathematical model allows us to represent and distinguish among several distinct worldviews toward allocation policies (see, e.g., \citep{sep-egalitarianism} for further discussion of these views): 
(1) a \emph{utilitarian} view, which generally favors slightly lower-ability members of $A$ to comparable, but slightly higher ability members of $D$ in pursuit of maximizing social utility and productivity (recall that membership in $A$ confers a boost in success probability); 
(2) a \emph{prioritarian} view, which evaluates a policy according to its impact on the well-being of the worse-off members of society. Our model can capture the priority view through large discount factors (recall that as the society's patience increases, it effectively increases the priority assigned to the disadvantaged group members), or by adjusting the welfare function; 
(3) a \emph{desert-principle} view, which advocates for allocating opportunities based on some notion of deservingness. Deservingness in this view is often defined in terms of the contributions people make to the social utility. Hence success probability in our model is arguably the closest match to individual desert. 
With that definition for desert, desert-based principles would allocate opportunities myopically in each generation. As our analysis illustrates, such policies often fail to maximize the social utility in the long-run.
}

\xhdr{Limitations and Interpretations}
Our model is designed to incorporate the basic points we just mentioned
in as simplified a fashion as possible; as such, it is important
to note some of its key limitations.
First, it is intended to model the effect of a single opportunity, and it treats other forms of mobility probabilistically in the background.
It also assumes that the fundamental parameters ($\alpha, \sigma, \tau, \gamma$) are constant over all generations \hl{as well as over individuals within one generation}.
It treats an individual's group membership ($A$ and $D$) and
ability as a complete description of their performance, rather than including any dependence on the group membership of the individual's parent. (That is, an individual in group $A$ performs the same in the model regardless of whether their parent belonged to group $A$ or $D$.)
All of these would be interesting restrictions to relax in an
extension of the model.
Second, much of the past theoretical work on intergenerational mobility focuses
on an issue that we do not consider here: the strategic considerations
faced by parents as they decide how much to consume in the present
generation and how much to pass on to their children.
Our interest instead has been in the optimization problem faced by
a social planner in allocating opportunities, treating the behavior of the agents as fixed and simple. Here too, it would be interesting to explore models that address these issues in combination.
Finally, because our focus is on intergenerational mobility in
a socioeconomic sense, we do not model discrimination based on race, ethnicity, or gender, and the role of race- or gender-based affirmative action in combatting these effects.
The model is instead concerned with
\emph{socio-economic} or \emph{class-based}~\citep{malamud1995class,kahlenberg1996class} affirmative action.
That said, the ingredients here could be combined with models of statistical or taste-based discrimination on these attributes to better understand their interaction \hl{(as outlined in Section~\ref{sec:conclusion})}.

The simplicity of our model, however, does allow us to make a correspondingly fundamental point: that even a purely payoff-maximizing society can discover affirmative action policies from first principles as
it seeks to optimize the allocation of opportunities over multiple 
generations.
Moreover, the optimal allocation policy is deeply connected
to dynamic programming over the generations;
the society is essentially attempting to ``steer'' the balance
of group $A$ and group $D$ over time, making sure not to turn things
too abruptly (giving up present benefit) or too gradually (giving up
future benefit).
This idea that society is searching for a way to turn optimally toward a better outcome is not specific to our model; it is an image 
that has arisen in qualitative discourse over several centuries.
It can be seen in a quote popularized by Martin Luther King,
that ``the arc of the moral universe is long, but it bends toward justice''
\citep{cohen-arc-of-moral-universe}.
Interestingly, the original form of this quote, by the American minister
Theodore Parker in 1853, has an even more abstractly mathematical flavor:
``I do not pretend to understand the moral universe; the arc is a long one, my eye reaches but little ways. I cannot calculate the curve
and complete the figure by the experience of sight; I can divine 
it by conscience. And from what I see I am sure it bends towards justice''
\citep{parker-arc-of-moral-universe}.
It is a curiously apt image for the way in which our optimal solutions 
gradually turn through the state space to
reshape the distribution of socioeconomic groups,
and it can be seen as added motivation for the issues at the heart of the model.

\subsection*{Related Work}\label{sec:related-brief}

Here, we briefly mention several lines of scholarship that are closely related to our work. See Appendix~\ref{sec:related} for a more in-depth discussion.

\xhdr{Long-term Implications of Fair ML} Several recent articles study the long-term impact of ML-based decision-making and fairness interventions on society, including the enforcement of statistical parity in hiring \citep{hu2018short}, and responses by individuals and populations to an ML-based decision rule \citep{liu2018delayed,mouzannar2019fair,kannan2018downstream}.
\citet{liu2018delayed}, for example, study the conditions under which the choices of a myopic profit-maximizing institution (e.g., a bank lending money to individuals) work in the interest of the disadvantaged group.
%
\cite{dong2018strategic,hu2018disparate,milli2018social} address \emph{strategic classification} where the goal is to design classifiers robust to strategic manipulation. 
This body of research focuses on strategic responses by agents being evaluated; in contrast, we assume idealized prediction so as to focus on
the perspective of a social planner who seeks optimal allocation of opportunities over time.

\xhdr{Intergenerational Income Mobility}
A substantial literature in economics studies how higher inequality results in lower income mobility across generations (see, e.g.,~\citep{becker1979equilibrium,maoz1999intergenerational,corak2013income}). The precise measurements of inequality and mobility significantly influence the strength of this effect (see, e.g.,~\citep{solon1992intergenerational,piketty2000theories}). 
Theoretical models of income mobility have studied utility-maximizing parents deciding how much of their capital to consume and how much of it to invest in their offspring~\citep{becker1986human,becker1979equilibrium,loury1981intergenerational,solon1999intergenerational}. We deliberately set aside parental strategic considerations and focus instead on deriving the optimal policy that maximizes the discounted payoff over generations.

\xhdr{Affirmative Action Policies}
A rich body of work in economics investigates statistical discrimination~\citep{arrow1973theory} and the role of affirmative action in redressing it \citep{fang2011theories}.
%
Outcome-based policies including affirmative action targets have long been proposed and implemented as {temporary} remedies to eliminate group-level inequalities. 
%
Race-based affirmative action and socioeconomic affirmative action
can be viewed as distinct categories of intervenions
\citep{kahlenberg1996class,carnevale2013socioeconomic},
with the former addressing long-term effects of racial bias and
the latter facilitating access
for economically disadvantaged individuals \citep{reardon2017can}.
While the relationship between them is complex and contested in the
literature \citep{kane1998racial,reardon2006implications,gaertner2013considering}, they can co-exist without necessarily competing.

\hl{
\xhdr{Affirmative Action in College Admissions and Comparison with \citep{durlauf2008affirmative}}
\citet{durlauf2008affirmative} provides a model to compare the equality and efficiency of affirmative action policies with those of meritocratic policies in the context of admission rules to public universities.
In his concluding remarks, Durlauf poses the key question our work sets out to answer: how do multigenerational considerations impact the optimal allocation policy? 
As we address this open question that he poses, we follow his model in many respects, although we depart from it in a few key areas.
\footnote{For a more detailed comparison, see Section~\ref{app:durlauf}.} Similar to our work, \citeauthor{durlauf2008affirmative} provides a condition under which efficiency and equality considerations are aligned, but he focuses on settings where a diverse body of students on college campuses improves the human capital development for all of them.  \citeauthor{durlauf2008affirmative} assumes the policymaker aims to maximize the level of human capital among \emph{the next generation} of adult citizens. He restricts attention to two generations only, and instead of solving for the optimal policy, he compares the meritocratic policy with affirmative action in terms of the average human capital of the next generation. In contrast, we model the dynamics of human-capital development across \emph{multiple generations} and derive the \emph{optimal allocation policy}. Moreover, groups in model corresponds to \emph{socio-economic tiers}, whereas \citeauthor{durlauf2008affirmative} defines them in terms of \emph{race}. 
}

\hl{
In both models, generations are represented as follows: in each time step, a new member is born into each family/dynasty, and he/she replaces the current member of the family in the next generation. 
The \emph{initial human capital} in Durlauf's model corresponds to our notion of \emph{success probability}. \emph{Adult human capital} in his model is determined by college attendance, and it roughly maps to our notion of \emph{success} (i.e., whether the individual succeeds if given the opportunity.) 
In both models, an admission rule maps a student's human capital and group membership into a binary outcome indicating whether the student is given the opportunity. In both models, the admission rule may vary across time and generations;
the level of state expenditures on education is assumed constant across generations ($\alpha$ is fixed in our model);
the only output of universities is human capital (our objective function is made up of the percentage of the population who succeed in each generation);
and finally, the initial human capital is accurately measured for every student (we assume ability and success probability are perfectly observable.)
}

\section{A Dynamic Model} 

\subsection*{Agents, Circumstances, Abilities, \& Generations}
We consider a model of a society that consists of a continuum of agents in 
two different sets of socioeconomic circumstances --- a {\em disadvantaged}
circumstance $D$ and an {\em advantaged} circumstance $A$.
These circumstances are (probabilistically) inherited from one generation to the next,
but we can try to increase the number of agents in the advantaged
circumstance in future generations by offering opportunities to disadvantaged
agents in the current generation.
This comes with a trade-off, however, since a competing option is
to offer these opportunities to advantaged agents in the current generation.
Our goal is to model this trade-off. 
We say that an agent $i$ has circumstance $c_i = 0$ 
if they are disadvantaged ($i \in D$), and 
circumstance $c_i = 1$ if they are advantaged ($i \in A$).
Each agent $i$ also has an ability $a_i$, which is a real number in $[0,1]$.

Time advances in discrete periods, beginning with period $t = 0$.
We think of these as generations. Consider an agent $i$ who has circumstance $c^{\text{init}}_i$ at the beginning of time $t$. Depending on whether $i$ receives the opportunity, his/her circumstance may change to $c_i^{\text{post}}$. 
At the end of time step $t$, $i$ produces a new agent $i'$ in generation $t+1$.
This new agent $i'$ has an ability $a_{i'}$ drawn uniformly at random from $[0,1]$.
With some fixed probability (specified below) $i'$ inherits the ex-post circumstance of $i$ (so $c_{i'} = c^{\text{post}}_i$), otherwise, it takes on a circumstance randomly selected from the background distribution of circumstances within the population in generation $t$. 
More specifically, in a given period $t$, let $\phi_j(t)$ denote the fraction of agents who 
have circumstance $j$, for $j=0,1$.
If $c^{\text{post}}_i = 0$, then with a fixed probability $1-p_D$, $i'$ inherits circumstance $D$, and with fixed probability $p_D$, it receives a circumstance randomly selected from the background distribution $(\phi_0(t),\phi_1(t))$. 
Similarly, If $c^{\text{post}}_i = 1$, then with a fixed probability $1-p_A$ , $i'$ inherits circumstance $A$, and with probability $p_A$, it receives a circumstance randomly selected from the background distribution $(\phi_0(t),\phi_1(t))$.


The movement probabilities, $p_A$ and $p_D$, capture all  processes---other than the opportunity we seek to allocate optimally---through which individuals can change their circumstance from their parental inheritance. For example, in the college admissions example, while our model focuses on how admission decisions can reshape circumstances over generations, there are many other forces and processes that impact the evolution of circumstances within society (e.g., number of jobs in the economy, training opportunities outside college, or pure luck). The movement probabilities summarize and capture all these alternative upward or downward movement possibilities.

\subsection*{Opportunities and Payoffs}
We consider the problem of performing an intervention in this society,
which consists of offering an {\em opportunity} to a subset of the population.
We only have the resources to offer the opportunity to an $\alpha$
fraction of the population.
An agent who is offered the opportunity has some probability of
succeeding at it, as a function of their ability and circumstances
that we specify below.
Succeeding at the opportunity confers two benefits on society:
\begin{enumerate}
\item[(i)] it produces an immediate payoff/reward to the society in the form of productivity;
\item[(ii)] if the agent is disadvantaged, it moves them (and subsequently their future generations) into the advantaged group.
\end{enumerate}
The central problem to be solved in the model, as we will see below,
is how to balance the immediate gains from (i) against the
long-term gains from (ii) over multiple generations.

In particular, if an agent of ability $a_i \in [0,1]$ 
and circumstance $c_i \in \{0,1\}$ is offered the opportunity,
their probability of succeeding at it is
$a_i \sigma + c_i \tau,$ 
where $\sigma, \tau > 0$ and
$\sigma + \tau \leq 1$.
Note that since $c_i \in \{0,1\}$, this
simply means that $\tau$ gets added to the success probability
of all agents whose circumstance is equal to $1$.

Our payoff (or reward) $r(t)$ in period $t$ is 
simply the fraction of the population that both
receives the opportunity and succeeds at it.
Our total payoff is a discounted sum of payoffs over all periods,
with discount factor $0< \gamma < 1$; that is, the total payoff
$r$ is equal to $\sum_{t = 0}^\infty \gamma^t r(t)$. 
As noted earlier, agents with circumstance $0$ who receive the
opportunity and succeed at it will produce offspring who (are more likely to) have
circumstance $1$; this matters for the payoff because the total payoff
$r = \sum_{t = 0}^\infty \gamma^t r(t)$ 
depends on the fraction of agents with each type of circumstance
in all time periods.

\subsection*{Thresholds and Interventions}
The way we allocate the opportunity at time $t$ is to 
set a threshold $\theta_j(t)$ for agents with circumstance $j$,
for $j \in \{0,1\}$, and to offer the opportunity to all agents 
$i$ with circumstance $j$ whose \textbf{success probability}, $a_i \sigma + c_i \tau$, is at least $\theta_j(t)$. That is,
$$a_i \sigma + c_i \tau \geq \theta_j.$$
We will sometimes write the threshold $\theta_j$ and the population
fraction with each circumstance $\phi_j$ 
without the explicit dependence ``$(t)$'' 
when we are considering a single fixed time period.

Agents of circumstance $0$ make up a $\phi_0$ fraction
of the population, and a $1 - \frac{\theta_0}{\sigma}$ fraction of them 
receive the opportunity, for a total fraction of the population
equal to $\phi_0 \times (1 - \frac{\theta_0}{\sigma})$.
Similarly, agents of circumstance $1$ make up a $\phi_1$ fraction
of the population, and a $1 - \frac{\theta_1 - \tau}{\sigma}$ fraction of them 
receive the opportunity, for a total fraction of the population
equal to $\phi_1 \times  (1 - \frac{\theta_1 - \tau}{\sigma})$.
The sum of these two fractions must add up to $\alpha$ --- the
portion of the population to whom we can offer the opportunity:
{\scriptsize
\begin{eqnarray}
&&\forall 0 \leq\theta_0 \leq \sigma \text{ and } \tau \leq \theta_1\leq \sigma+\tau: \text{   }\phi_0 \times  \left(1 - \frac{\theta_0}{\sigma}\right) + \phi_1 \times  \left(1 - \frac{\theta_1 - \tau}{\sigma}\right) = \alpha \nonumber \\
&\Leftrightarrow &\forall 0 \leq\theta_0 \leq \sigma \text{ and } \tau \leq \theta_1\leq \sigma+\tau: \text{   }
\phi_0 \theta_0 + \phi_1 \theta_1 = \sigma (1 - \alpha)+ \phi_1\tau.
\label{eq:capacity-constraint}
\end{eqnarray}
}
This also shows how our choice of thresholds is a one-dimensional
problem in the single variable $\theta_0$ (or equivalent 
in the single variable $\theta_1$), since after setting one of the
two thresholds, the other is determined by this equation. 
More precisely, we have that:
\begin{equation}\label{eq:theta_1}
\begin{cases}
\theta_1 = \frac{\sigma(1-\alpha) + (1-\phi_0)\tau - \phi_0 \theta_0}{1-\phi_0} \quad \forall \phi_0<1\\
\theta_1 \in [\tau,\sigma+\tau] \quad \text{ for } \phi_0 = 1.
\end{cases}
\end{equation}
Note that if $\phi_0=1$, $\theta_1$ can take on any value in $[\tau,\sigma+\tau]$, the threshold will not affect how opportunities are allocated. The same holds for $\phi_0=0$ and any $\phi_0 \in [0,\sigma]$.

\subsection*{Dynamics in a Single Period}
Recall that our payoff $r(t)$
in period $t$ is simply the fraction of the population
that both receives the opportunity and succeeds at it.
We can decompose this as follows.
\begin{itemize}
\item 
Agents of circumstance $0$ make up a $\phi_0$ fraction
of the population, and a $1 - \frac{\theta_0}{\sigma}$ fraction of them 
receive the opportunity, for a total fraction of the population
equal to $\phi_0 \left(1 - \frac{\theta_0}{\sigma} \right)$.
Not all of these agents succeed at the opportunity;
the average success probability in this group is 
$\frac12 (\sigma + \theta_0)$, so the expected
quantity that succeeds is 
$$ \phi_0 \left(1 - \frac{\theta_0}{\sigma}\right) \frac{(\sigma + \theta_0)}{2} =
  \frac{\phi_0}{2\sigma} \left(\sigma^2 - \theta_0^2 \right).$$
\item 
Agents of circumstance $1$ make up a $\phi_1$ fraction
of the population, and a $1 - \frac{\theta_1 - \tau}{\sigma}$ fraction of them 
receive the opportunity, for a total fraction of the population
equal to $\phi_1 \left(1 - \frac{\theta_1 - \tau}{\sigma} \right)$.
Again, not all of these agents succeed at the opportunity;
the average success probability in this group is 
$\frac12 (\sigma+\tau + \theta_1)$, so the expected
quantity that succeeds is 
$$\phi_1 \left(1 - \frac{\theta_1 - \tau}{\sigma} \right) \frac{(\sigma+\tau + \theta_1)}{2} = 
\frac{\phi_1}{2\sigma} \left((\sigma+\tau)^2 - \theta_1^2\right).
$$
\end{itemize}
The total payoff in period $t$ is the sum of these two terms:
\begin{equation}
r(t) =   \frac{\phi_0(t)}{2\sigma} \left(\sigma^2 - \theta_0(t)^2 \right)
+ \frac{\phi_1(t)}{2\sigma} \left((\sigma+\tau)^2 - \theta_1(t)^2\right).
\label{eq:period-t-payoff}
\end{equation}

\subsection*{Dynamics over Multiple Periods}
If we were just optimizing the payoff in this single time period,
then we'd have a single-variable optimization problem in the variable
$\theta_0$ (or equivalently in $\theta_1$), with the objective function
given by (\ref{eq:period-t-payoff}).
But since there is also the set of discounted payoffs in future time periods,
we also need to look at the effect of our decisions
on the quantities $\phi_0(.)$ and $\phi_1(.)$ in future periods.

If $p_A = p_D = 0$, $\phi_1(t+1)$ grows relative to $\phi_i(t)$ depending on 
the fraction of the population that transitions from circumstance $0$
to circumstance $1$ by succeeding at the opportunity. Thus we have
\begin{equation}
\phi_0(t+1) = \phi_0(t) 
  - \frac{\phi_0(t)}{2\sigma} \left(\sigma^2 - \theta_0(t)^2 \right).
\label{eq:period-t-transition}
\end{equation}
More generally when $p_A$ or $p_D$ are non-zero, let's define $\phi^{\text{post}}_0 (t)  = \phi_0(t) 
  - \frac{\phi_0(t)}{2\sigma} \left(\sigma^2 - \theta_0(t)^2 \right)$. (For simplicity, we drop ``$(t)$'' and simply use $\phi^{\text{post}}_0$ in the remainder of this section). We have:
\begin{equation}
\phi_0(t+1)  = \phi^{\text{post}}_0  (1-p_D) + \phi^{\text{post}}_0  p_D \phi^{\text{post}}_0  + \left(1-\phi^{\text{post}}_0\right) p_A \phi^{\text{post}}_0,
\label{eq:period-t-transition-general}
\end{equation}
It is easy to see that:
\begin{proposition}
If $p_A = p_D$, then $\phi_0(t+1) = \phi_0(t) 
  - \frac{\phi_0(t)}{2\sigma} \left(\sigma^2 - \theta_0(t)^2 \right)$.
\end{proposition}
\begin{proof}
Suppose $p_A = p_D = p$. Then we can re-write (\ref{eq:period-t-transition-general}) as follows:
\begin{eqnarray*}
\phi_0(t+1)  &=& \phi^{\text{post}}_0 (1-p) + \phi^{\text{post}}_0 p \phi^{\text{post}}_0 + \left(1-\phi^{\text{post}}_0 \right) p \phi^{\text{post}}_0 \\
&=& \phi^{\text{post}}_0 (1-p) + \phi^{\text{post}}_0 p \left( \phi^{\text{post}}_0  + \left(1-\phi^{\text{post}}_0 \right) \right) \\
&=& \phi^{\text{post}}_0 (1-p) + \phi^{\text{post}}_0 p  \\
&=& \phi^{\text{post}}_0 = \phi_0(t)  - \frac{\phi_0(t)}{2\sigma} \left(\sigma^2 - \theta_0(t)^2 \right)
\end{eqnarray*}
where in the last line, we replace $\phi^{\text{post}}_0$ with its definition.
\end{proof} 
The above proposition shows that with respect to dynamics and optimal policy, settings in which $p_A = p_D$, are essentially equivalent to settings in which $p_A = p_D = 0$.
  
In summary, the full problem is to choose thresholds $\theta_0(t), \theta_1(t)$
for each time period $t$ so as to maximize the infinite sum 
$r = \sum_{t = 0}^\infty \gamma^t r(t)$.
Each term $r(t)$ depends not just on the chosen thresholds 
but also on the fractions of agents with each type of circumstance
$\phi_0(t), \phi_1(t)$, which evolve according to the recurrence in 
Equation (\ref{eq:period-t-transition-general}).
Note that the intuitive trade-off between $\theta_0$ and $\theta_1$
shows up in the formulation of 
(\ref{eq:period-t-payoff}) and (\ref{eq:period-t-transition}):
lowering $\theta_0$ and lowering $\theta_1$ have different effects,
both in period $t$ and in future time periods.

\section{Theoretical Analysis for $p_D \geq p_A$}

In this section, we focus on settings in which $p_D \geq p_A$ and characterize the optimal policy to maximize the discounted payoff over generations. (We only provide the analysis for settings of $p_D = p_A$ but the extension to $p_D>p_A$ is straightforward). We cast the problem as deriving the infinite-time-horizon optimal policy in a continuous state- and action-space (Markov) decision process.
We characterize the optimal threshold and value function for every state $\phi_0 \in [0,1]$. Importantly, we show that there exists a tipping point $\phi^*_0$ that splits the state space into two distinct regions: states at which the optimal threshold uses strict affirmative action, and states at which the optimal policy consists of imposing equally high thresholds on both $A$ and $D$ groups.

\subsection*{The Decision Process}\label{sec:DP}

Given $\alpha, \sigma, \tau$, and $\gamma$, we define a decision process $\cD_{\alpha, \sigma, \tau,\gamma} =(\Phi, \Theta, S, R)$ (or $\cD$ for short) with a continuous state space $\Phi=[0,1]$, action space $\Theta=[0,\sigma]$, state transition $S: \Phi \times \Theta \rightarrow \Phi$, and reward function $R: \Phi \times \Theta \rightarrow [0,1]$. Each state $\phi_0 \in \Phi$ corresponds to a particular fraction of disadvantaged individuals within the population. For instance, the states $0$ (or $1$) represents a society in which no one (or everyone) belongs to $D$. 

The set of thresholds admissible in each state $\phi_0$ is denoted by $\Theta_{\phi_0}$.  $\Theta_{\phi_0}$ consist of all thresholds $0 \leq \theta_0 \leq \sigma$ that can satisfy the capacity constraint, $\alpha$. In other words, for any $ \in \theta_0\Theta_{\phi_0}$ if we impose the threshold $\theta_0 \in [0,\sigma]$ on group $D$, we can find a threshold $\theta_1 \in [0,\sigma+\tau]$ for group $A$ such that exactly $\alpha$ fraction of the overall population receives the opportunity. This capacity constraint translates into two conditions on $\Theta_{\phi_0}$: 
\begin{enumerate}
\item A threshold $\theta_0 \in \Theta_{\phi_0}$ should not give the opportunity to \emph{more than} $\alpha$ fraction of the population. Formally,
$\forall \theta_0 \in \Theta_{\phi_0}:  \phi_0 \left(1-\frac{\theta_0}{\sigma}\right) \leq \alpha$, which is equivalent to:
\begin{equation}\label{eq:action_lb}
\forall \theta_0 \in \Theta_{\phi_0}: \quad \theta_0 \geq \sigma \left( 1 - \frac{\alpha}{\phi_0} \right).
\end{equation}
\item A threshold $\theta_0 \in \Theta_{\phi_0}$ should not waste opportunities, that is, it should give the opportunity to \emph{at least} $\alpha$ fraction of the overall population. Formally,
$\forall \theta_0 \in \Theta_{\phi_0}:   \phi_0 \left(1-\frac{\theta_0}{\sigma}\right) + \phi_1 \geq \alpha.$
Replacing $\phi_1$ with $1-\phi_0$ and rearranging terms, the above is equivalent to
\begin{equation}\label{eq:action:ub}
\forall \theta_0 \in \Theta_{\phi_0}: \quad  \theta_0 \leq \frac{\sigma(1-\alpha)}{\phi_0}.
\end{equation}
\end{enumerate}
Figure~\ref{fig:actions} illustrates the actions satisfying conditions (\ref{eq:action_lb}) and (\ref{eq:action:ub}) for every state $\phi_0 \in [0,1]$.

\begin{figure}
\centering
    \includegraphics[width=0.4\textwidth]{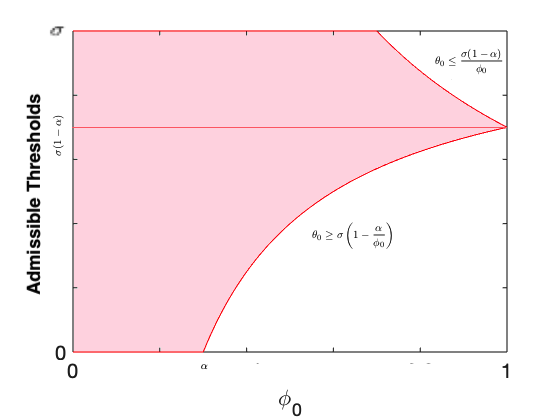}
    \caption{
       The admissible thresholds for every state of the decision process, $\cD$. The x-axis specifies the state $\phi_0$, and the y-axis highlights (in light red) the admissible thresholds, $\Theta_{\phi_0}$, at every state, $\phi_0$.
    } \label{fig:actions}
\end{figure}


The state transition function, $S: \Phi \times \Theta \rightarrow \Phi$, specifies the state transitioned to for every (state, admissible threshold) pair. Formally,
\begin{equation}\label{eq:state}
\forall \phi_0 \in \Phi, \forall \theta_0 \in \Theta_{\phi_0}: \quad  S(\phi_0, \theta_0) = \phi_0 - \frac{\phi_0}{2\sigma} (\sigma^2 - \theta_0^2).
\end{equation}
The reward function, $R: \Phi \times \Theta \rightarrow [0,1]$, is defined as follows: $R(\phi_0, \theta_0)$ denotes the immediate reward/payoff of imposing threshold $\theta_0$ at state $\phi_0$. Formally,
$$\forall \phi_0 \in \Phi, \forall \theta_0 \in \Theta_{\phi_0}: \quad   R(\phi_0, \theta_0) = \frac{\phi_0}{2\sigma} (\sigma^2 - \theta_0^2)
+ \frac{\phi_1}{2\sigma} \left((\sigma+\tau)^2 - \theta_1^2\right).$$
Replacing $\phi_1$ with $1-\phi_0$ and $\theta_1$ with the right hand side of (\ref{eq:theta_1}), we obtain the following equivalent expression for $R$:
{\small
\begin{equation}\label{eq:reward}
R(\phi_0, \theta_0) =  
\begin{cases}
\frac{1}{2\sigma} (\sigma^2 - \theta_0^2) \quad \text{ for } \phi_0 = 1 \text{, otherwise: }\\
\frac{\phi_0}{2\sigma} (\sigma^2 - \theta_0^2)
+ \frac{1-\phi_0}{2\sigma} \left((\sigma+\tau)^2 - \left( \frac{\sigma(1-\alpha) + (1-\phi_0)\tau - \phi_0 \theta_0}{1-\phi_0}\right)^2\right) 
\end{cases}
\end{equation}
}

\subsection*{Characterization of the Optimal Policy} 
Next, we illustrate and characterize the optimal policy for the decision process $\cD$ defined above. (For further information and references on continuous-state decision processes, see Appendix~\ref{app:continuous}.)

A deterministic policy $\pi$ for a decision process $\cD$ is a mapping $\pi: \Phi \rightarrow \Theta$
such that $\pi(\phi_0)$ prescribes the threshold at state $\phi_0$. The value $V_{\pi}(\phi_0)$ of a state $\phi_0$ under policy $\pi$ is the discounted reward of executing policy $\pi$ on $\cD$ starting with initial state $\phi_0$. A policy $\pi$ is optimal if its value function $V_\pi$ satisfies Bellman Optimality---defined recursively as follows:
\begin{equation}\label{eq:bellman}
V_{\pi}(\phi_0) = \max_{\theta_0} R(\phi_0, \theta_0) + \gamma V_{\pi}(S(\phi_0, \theta_0)).
\end{equation}
We establish in Appendix~\ref{app:value} that for our decision process $\cD$, the value function satisfying the above functional equation is \emph{unique}, \emph{continuous}, 
and \emph{differentiable}. 
(We prove these facts utilizing tools from recursive analysis and dynamic programming~\citep{stokey1989recursive,cotter2006non}). 
For simplicity, from this point on we refer to this unique optimal value function as $V(.)$ and drop the subscript $\pi$.

Let the correspondence $\Pi^*_0: \Phi \rightarrow 2^\Theta$ denote all optimal policies for $\cD$. More precisely, for any $0 \leq \phi_0 \leq 1$, the set $\Pi^*_0(\phi_0)$ contains all optimal threshold values at $\phi_0$.  
Figure~\ref{fig:policy_brief} illustrates $\Pi^*_0$ for a sample setting of the parameters $\alpha, \sigma, \tau, \delta$. (See Figure~\ref{fig:policy} in the Appendix for more instances.) Figure~\ref{fig:value} in the Appendix illustrates the value function $V$. Note that the value function is consistently decreasing and concave.

 \begin{figure*}[t!]
\centering
\begin{subfigure}[t]{0.31\textwidth}
  \centering
          \includegraphics[width=\textwidth]{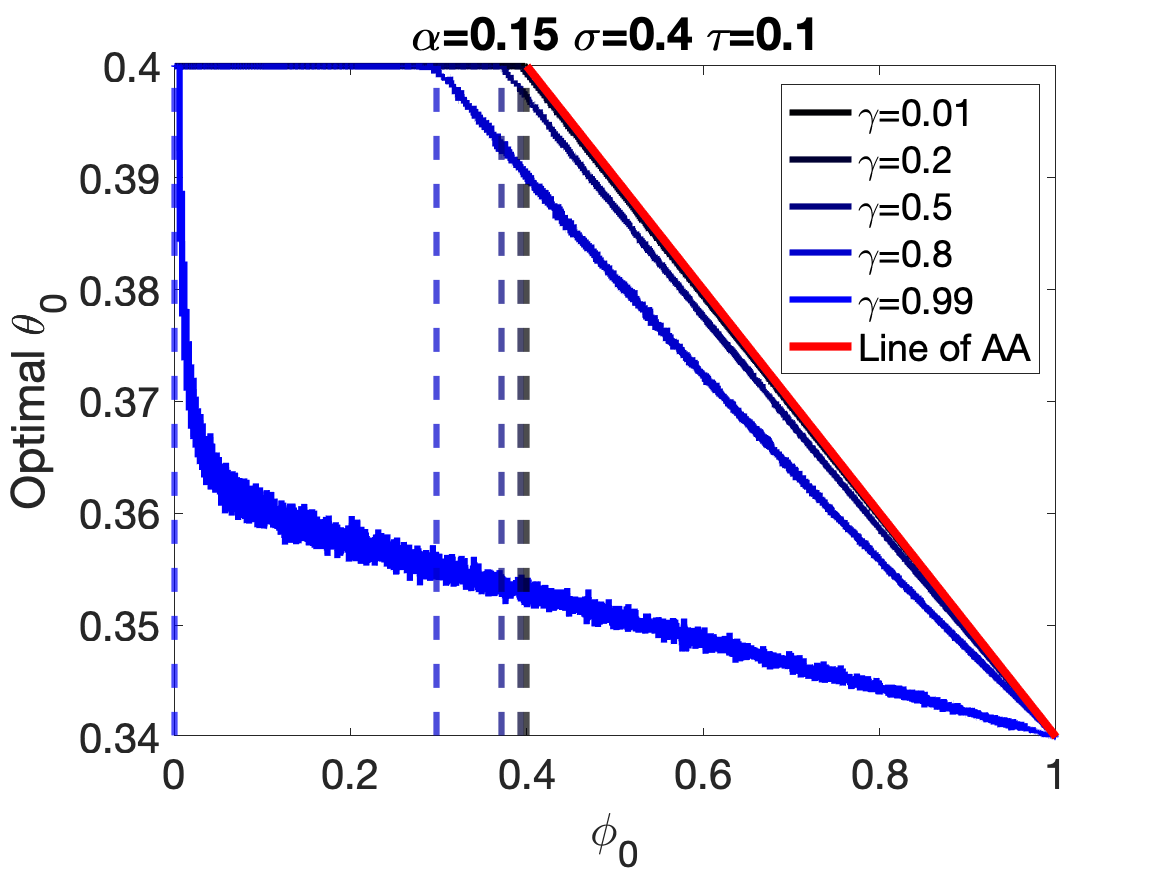}
   \caption{The optimal $\theta_0$ at every state $0 \leq \phi_0 \leq 1$. \textbf{The optimal threshold decreases with $\phi_0$}. The dashed lines indicate the tipping points, $\phi^*_0$, below which $\sigma$ is the only optimal threshold, and above it $\sigma$ is not optimal.}\label{fig:policy_brief}
\end{subfigure}\hspace{2mm}
\begin{subfigure}[t]{0.31\textwidth}
  \centering
          \includegraphics[width=\textwidth]{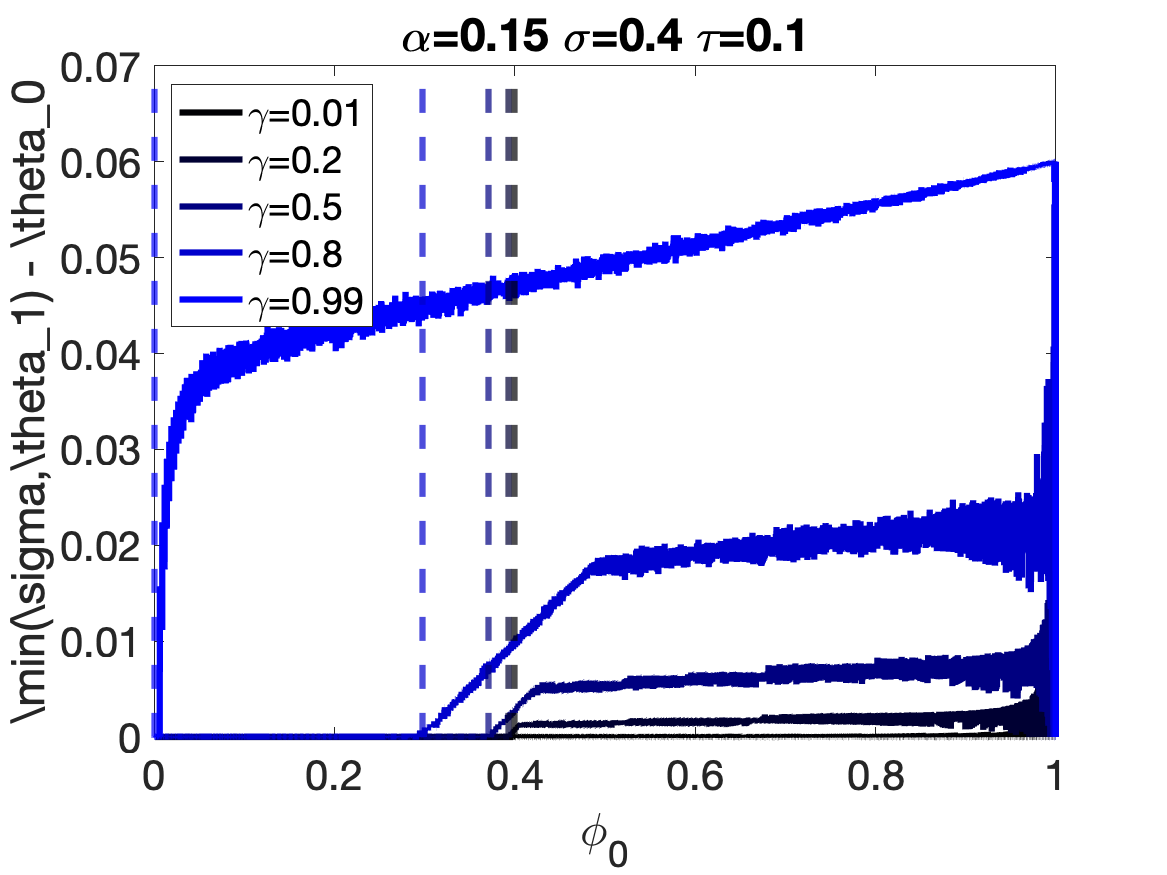}
        \caption{The difference between $\theta_0$ and $\theta_1$ at every state $0 \leq \phi_0 \leq 1$. The dashed lines specify the tipping points. Strict affirmative action is employed beyond $\phi^*_0$ only and \textbf{the extent of affirmative action is increasing in $\phi_0$}.}\label{fig:AA_brief}
\end{subfigure}\hspace{2mm}
\begin{subfigure}[t]{0.31\textwidth}
  \centering
          \includegraphics[width=\textwidth]{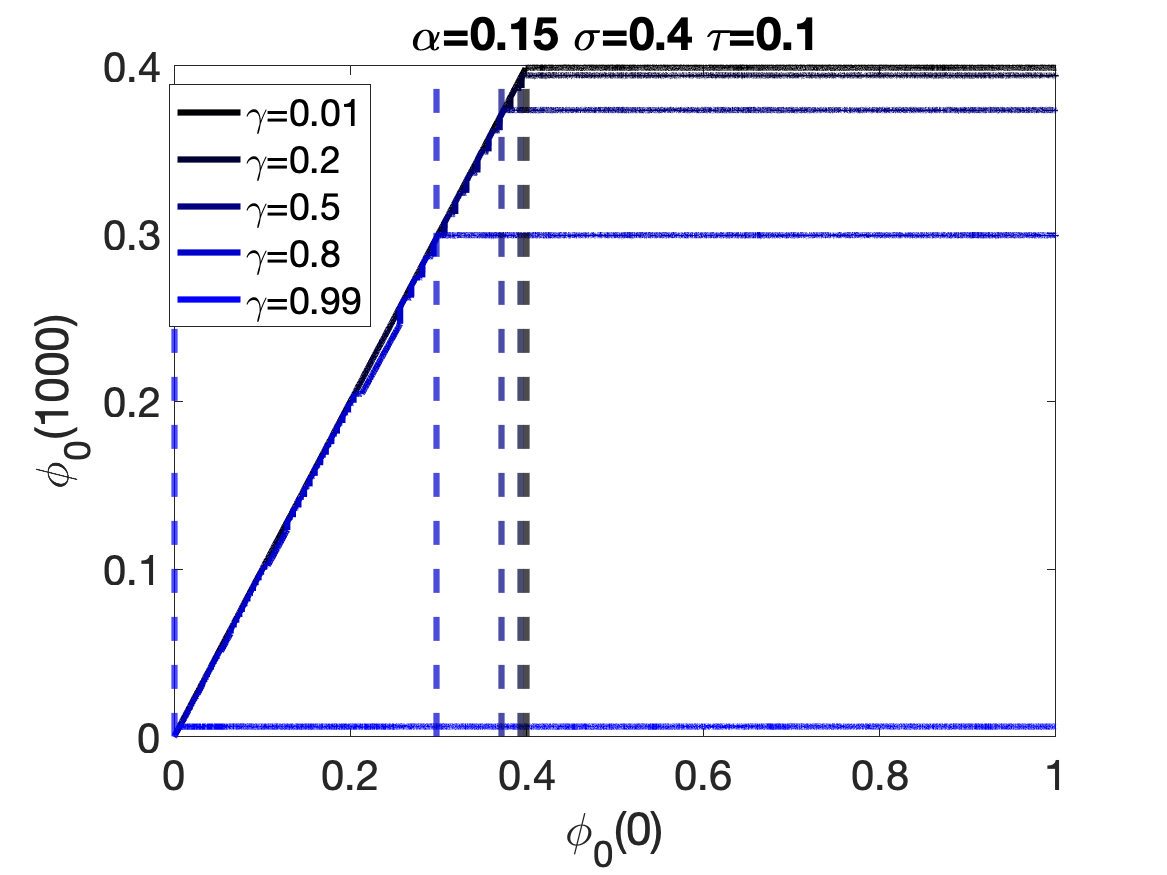}
        \caption{The state to which the optimal policy converges given the initial state $\phi_0$. The dashed lines specify the tipping point, $\phi^*_0$. Note that \textbf{the optimal policy never shrinks the size of group $D$ to a value less than $\phi^*_0$}.}\label{fig:convergence_brief}
\end{subfigure}
\caption{Illustration of the (a) optimal policy, (b) extent of affirmative action, and (c) absorbing state.}
\end{figure*}

We say that the optimal policy, $\Pi^*_0$, uses affirmative action at a state $\phi_0$ if at $\phi_0$, it imposes a lower threshold on group $D$ compared to $A$. 
\begin{definition}[(Strict) Affirmative Action]
The optimal policy uses (strict) affirmative action at state $\phi_0$ if for all $\theta_0 \in \Pi^*_0(\phi_0)$ and all $\theta_1 \in \Pi^*_1(\phi_0)$, $\theta_0 < \theta_1$.
\end{definition}
If the inequality above is not strict (i.e., $\theta_0 \leq \theta_1$), we say the optimal policy uses \emph{weak} affirmative action.
%
Figure~\ref{fig:AA_brief} shows the extent of affirmative action for a sample setting of the parameters $\alpha, \sigma, \tau, \delta$. (See Figure~\ref{fig:AA} in the Appendix for more instances).

Our main result (Theorem~\ref{thm:main}) characterizes the states at which the optimal policy employs affirmative action. In particular, it establishes the existence of a tipping point $0 \leq \phi^*_0 \leq 1-\alpha$ that designate the region of affirmative action: At any state $\phi_0 \leq \phi^*_0$, the optimal policy assigns similar high threshold to both the advantaged and disadvantaged groups. In contract, at any state $\phi_0 \leq \phi^*_0$ the optimal policy imposes a strictly lower threshold on group $D$ compared to $A$.
\begin{theorem}\label{thm:main}
Given $\alpha, \sigma, \tau, \gamma$ and the decision process $\cD$, let 
\begin{equation}\label{eq:phi_star}
\phi^*_0 =  \max \left\{ 0, \min \left\{ 1-\alpha , 1 - \frac{\alpha \sigma}{2\tau} \left( 1 + \sqrt{1 + \frac{2\tau\gamma}{1-\gamma}} \right) \right\} \right\}.
\end{equation} 
\begin{itemize}
\item For any state $\phi_0 \leq \phi^*_0$, there exists $\theta_0 \in \Pi^*_0(\phi_0)$ and $\theta_1 \in \Pi^*_1(\phi_0)$ such that $\theta_0 = \theta_1$.
\item For any state $\phi_0 > \phi^*_0$, for all $\theta_0 \in \Pi^*_0(\phi_0)$ and all $\theta_1 \in \Pi^*_1(\phi_0)$ such that $\theta_0 < \theta_1$.
\end{itemize}
\end{theorem}

We prove Theorem~\ref{thm:main} by establishing a series of Lemmas.
Lemma~\ref{lem:phi_0_star} determines the largest state $\phi^*_0 \in [0,1]$ below which the optimal policy consists of applying the high threshold of $\sigma$ to group $D$. Clearly, below such point, the optimal policy does not use affirmative action.  
Next, we investigate states larger than $\phi^*_0$. 
For every state $\phi_0 > \phi^*_0$, Lemma~\ref{lem:AA-line} identifies the set of thresholds that exhibit affirmative action. 
Lemma~\ref{lem:AA-leq} establishes that the optimal policy uses weak affirmative action beyond $\phi^*_0$. That is, it shows that for any state $\phi_0 > \phi^*_0$, it is never optimal to impose a strictly higher threshold on D compared to A.
Proposition~\ref{prop:strict-AA} shows that beyond $\phi^*_0$, the optimal policy in fact uses \emph{strict} affirmative action. That is, at every state $\phi_0 > \phi^*_0$, the optimal policy imposes a strictly lower threshold on D compared to A. 

Lemma~\ref{lem:phi_0_star} determines the state $\phi^*_0 \in [0,1]$ up to which $\sigma$ is an optimal threshold. Note that if $\sigma$ is an optimal threshold at a state $\phi_0$, the optimal policy does not use affirmative action at $\phi_0$. To see this, note that when $\sigma$ is optimal, any action $\theta_0 > \sigma$ is also optimal (they all pick a 0-fraction of group $D$ and are, therefore, effectively equivalent). 
(All omitted proof can be found in Appendix~\ref{app:omitted_proofs}.)

\begin{lemma}[The Tipping Point]\label{lem:phi_0_star}
For any state $\phi_0 < \phi^*_0$, $\sigma \in \Pi^*_0(\phi_0)$.
\end{lemma}

The following Lemma characterizes the region of affirmative action in the state-action space.
\begin{lemma}[Region of Affirmative Action]\label{lem:AA-line}
For a state $\phi_0^* < \phi_0 <1$, the threshold $\theta_0 \in \Pi^*_0(\phi_0)$ uses affirmative action if and only if $\theta_0 \leq \sigma(1-\alpha)+\tau(1-\phi_0)$.
\end{lemma}
\begin{proof}
Let $\theta$ be the threshold that if applied to both D and A at $\phi_0$, allocates opportunities exactly. We have that
$$\phi_0 \left(1-\frac{\theta}{\sigma} \right) + (1-\phi_0) \left(1-\frac{\theta-\tau}{\sigma} \right) = \alpha$$
or equivalently,
$$\theta = \sigma(1-\alpha)+\tau(1-\phi_0).$$
Note that $\theta_0 < \theta$ if and only if $\theta_1 > \theta$---otherwise the capacity constraints would not be maintained. Therefore, for any $\theta_0 < \theta$, $\theta_1>\theta_0$, which implies affirmative action.
\end{proof}

The following Lemma shows that beyond $\phi^*_0$, the optimal threshold for the disadvantaged is never higher than that for the advantaged.
\begin{lemma}[Weak Affirmative Action]\label{lem:AA-leq}
Consider a state $\phi_0 >\phi^*_0$.
For all $\theta_0 \in \Pi^*_0(\phi_0)$ and all $\theta_1 \in \Pi^*_1(\phi_0)$, $\theta_0 \leq \theta_1$.
\end{lemma}
\begin{proof}
Suppose not and there exists $\theta_0 \in \Pi^*_0(\phi_0)$ and $\theta_1 \in \Pi^*_1(\phi_0)$ such that $\theta_0 > \theta_1$. If we lower $\theta_0$ down to $\sigma(1-\alpha)+\tau(1-\phi_0)$ and increase $\theta_1$ up to $\sigma(1-\alpha)+\tau(1-\phi_0)$, we maintain the capacity constraints and at the same time achieve the following:
\begin{itemize}
\item[(a)] We improve the immediate reward of the current time step--- because we replace advantaged agents with low success probabilities in the range of $[\theta_1, \sigma(1-\alpha)+\tau(1-\phi_0)]$ with disadvantaged agents with higher success probabilities in $[\sigma(1-\alpha)+\tau(1-\phi_0), \theta_0]$.
\item[(b)] We move to a state with a relatively smaller size of group $D$ (simply because we gave the opportunity to more disadvantaged agents). The value function is strictly decreasing in $\phi_0$. Therefore, this new next state has a higher value compared to the previous one. 
\end{itemize}
The fact that we can improve the value contradicts the optimality of $\theta_0, \theta_1$. Therefore, $\theta_0 > \theta_1$ cannot hold
\end{proof}

The following Proposition establishes that beyond $\phi^*_0$, the optimal policy uses strict affirmative action. (The proofs establishing the proposition can be found in Appendix~\ref{app:omitted_proofs}.)
\begin{proposition}[Strict Affirmative Action]\label{prop:strict-AA}
Consider a state $\phi_0 >\phi^*_0$, and let $V'(\phi_0)$ denote the derivative of the value function $V$ evaluated at $\phi_0$. If $V'(\phi_0) < 0$, then for all $\theta_0 \in \Pi^*_0(\phi_0)$ and all $\theta_1 \in \Pi^*_1(\phi_0)$, $\theta_0 < \theta_1$.
\end{proposition}

\xhdr{Insights from the Analysis}
We end this section by making several observations about the optimal policy:
First, note that the optimal policy never shrinks the size of $D$ to a value less than $\phi^*_0$. For every initial state $\phi_0 \in [0,1]$, Figure~\ref{fig:convergence} shows the state one converges to when the optimal policy is simulated for $1000$ steps on $\cD$. In other words, affirmative action is optimal from a utilitarian point of view as long as group $D$ is sufficiently large. 
Second, the precise derivation of $\phi_0^*$, as specified in (\ref{eq:phi_star}), allows us to gain new insights into how the interaction between the parameters of our model can give rise to or avert affirmative action. Figure~\ref{fig:star} depicts the status of persistent affirmative action (i.e., $\phi_0^* \leq 0$) for $\alpha, \tau, \gamma$ in settings where $\sigma = 1-\tau$. (The derivation behind the plots can be found in Appendix~\ref{app:star}.) 
Notice that persistent affirmative action is promoted by large values of $\alpha$ and small values of $\tau$, since these make it easier to include 
high-performing members of group $D$ without a large difference in thresholds.
Persistent affirmative action is also promoted by larger values of $\gamma$,
indicating greater concern for the rewards in future generations. Note, however, that a finite level of patience often suffices for persistent affirmative action to be optimal.
%
%
\hl{
Finally, a frequent objection to affirmative action polices is their potential for reverse discrimination (see, e.g., Hopwood v. Texas, 78 F.3d 932 (5th Cir. 1996)). Translating these concerns into the terminology of our model, one may object that ``if the highest-ability member of $D$ has ability below the lowest-ability member of $A$, then \emph{any} affirmative action in this scenario will violate the desert principle''. Note, however, that in our model, abilities for both $A$ and $D$ group members are uniform in the $[0,1]$ interval. So  the optimal policy will never favor a lower-ability member of $D$ to a higher ability member of $A$. 
}



\section{Computational Analysis for $p_D < p_A$}
The focus of our theoretical analysis was on the settings in which
$p_A \leq p_D$. Next, we computationally investigate the optimal policy
for cases where $p_A > p_D$.

Recall that when $p_A$ and $p_D$ are non-zero, the circumstance of an
offspring is not deterministically specified by that of their parent.
Instead, $p_A$ and $p_D$ specify the offspring's probability of
spontaneous movement to the background distribution of circumstances
in society, that is, $(\phi_0,\phi_1)$. When $p_A \leq p_D$, both the
spontaneous movement dynamics and the allocation of opportunities work
toward reducing the relative size of group $D$ over time. This fact
allowed us to utilize backward induction to characterize the optimal
policy theoretically. 
When $p_A > p_D$, however, the spontaneous movement work in the
opposite direction of opportunity allocation: with no opportunity
allocated, the spontaneous movement dynamics gradually shift the
entire population to group $D$. In such settings, the role of the
allocation policy consists of combatting the natural flow of the
population to $D$.

Although in settings with $p_A > p_D$, it is significantly more
challenging to characterize the optimal allocation policy, we can
still approximately compute this policy via discretization followed by
solution methods for finite decision problems.  More
specifically, we discretize the state- and action-space of the
decision process, $\cD$, then compute the optimal policy of the
resulting finite decision process using methods such as
policy iteration. In what follows, we report the result of our
simulations when the state- and action-space are approximated by 1000
equidistant states and actions (i.e., $\tilde{\Phi} = (0, 0.001,
0.002, \cdots, 1)$ and $\tilde{\Theta} = \sigma (0, 0.001, 0.002,
\cdots, 1)$) and the reward and transition functions are approximated
accordingly.

    \begin{figure}
    \centering
        \includegraphics[width=0.4\textwidth]{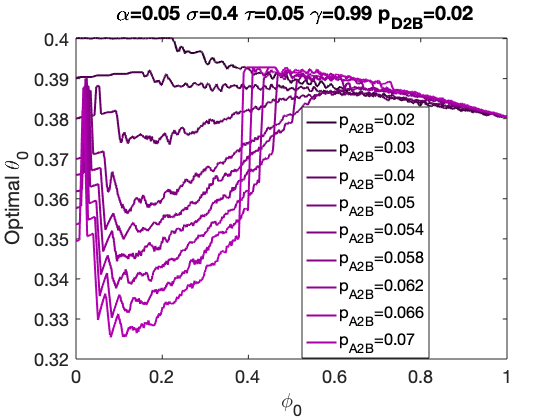}
        \caption{The optimal policy $\theta_0$ at every state $0 \leq \phi_0 \leq 1$ for various $p_A$ values.}
        \label{fig:p_A2B_policy}
    \end{figure}

    \begin{figure}[h]
    \centering
        \includegraphics[width=0.4\textwidth]{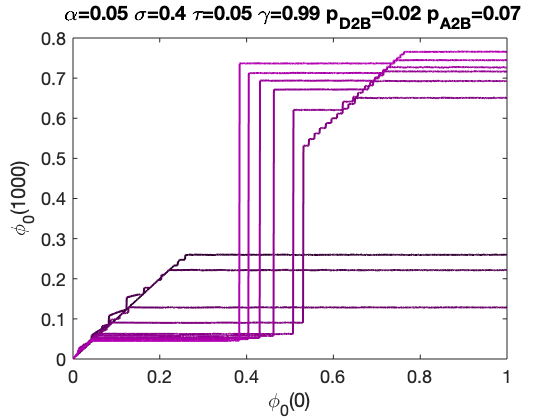}
        \caption{The absorbing state for every initial state $0 \leq \phi_0 \leq 1$. \textbf{The optimal policy is non-monotone and exhibits sharp drops when better absorbing states become cost-effective.}}
        \label{fig:p_A2B_absorbing}
    \end{figure}

Figure~\ref{fig:p_A2B_policy} shows how the optimal policy varies with $p_A$ for a fixed value of $p_D$. Compared to the optimal policy for settings of $p_A \leq p_D$, we observe two significant differences here:
\begin{itemize}
\item The optimal policy can be \textbf{non-monotone} in $\phi_0$.
\item There exists a \textbf{sharp drop} in the optimal threshold at some intermediate state $\phi_0 \in [0,1]$.
\end{itemize}
Next, we provide some justifications for both the intuitive and 
potentially counter-intuitive aspects of these phenomena.

For every initial state $\phi_0 \in [0,1]$, Figure~\ref{fig:p_A2B_absorbing} illustrates the state one converges to (i.e., the absorbing state) if the optimal policy is simulated for 1000 steps. Note that in all cases, the optimal policy successfully reduces the absorbing state to a more desirable value less than 1 (recall that with no intervention, we always converge to $\phi_0=1$ through the spontaneous movements). 

As illustrated in Figure~\ref{fig:p_A2B_absorbing}, the sharp drop in the optimal threshold coincides with an abrupt change in the absorbing state. When the initial state is close to 1, the allocation policy must employ extensive affirmative action to reduce and maintain a better absorbing state. This feat, however, comes at the cost of immediate reward. The optimal policy, therefore, settles for reaching and obtaining a slightly more desirable absorbing state (a value between 0.6 and 0.8 in Figure~\ref{fig:p_A2B_absorbing}). 
For smaller initial states, the cost of employing extensive affirmative action goes down, and at some point, it becomes viable to reach and maintain a much better absorbing state (a value between 0 to 0.3 in Figure~\ref{fig:p_A2B_absorbing}). The sharp drop in the optimal threshold happens precisely at the state where the benefit of extensive affirmative action outweighs its cost. 

In summary, when $p_A > p_D$, we need persistent (albeit non-monotone)
affirmative action to reach and maintain more desirable absorbing
states. The optimal extent of affirmative action crucially depends on
the initial state of the population: If the disadvantaged group is
large to begin with, the optimal policy has to forgo extensive
affirmative action to maintain sufficient short-term return. If the
advantaged group initially has a sufficient mass, extensive
affirmative action becomes optimal, and in the long-run, manages to
significantly increases the fraction of advantaged individuals in the
population. These findings are in sharp contrast with settings of $p_A
\leq p_D$. In those settings, affirmative action is optimal under
the much more straightforward condition that the disadvantaged group
exceeds a $\phi_0^*$ fraction of the population, for a constant
$\phi_0^*$ that we can specify precisely, 
and it ceases to be optimal as soon as society gains at
least $1 - \phi_0^*$ mass in the advantaged group.

\section{Discussion and Future Directions}\label{sec:conclusion}
In this paper, we developed and analyzed a model for allocating
opportunities in a society with limited intergenerational mobility.
These opportunities produce benefits in the present generation, but
they can also raise the socioeconomic status of the recipients
and their descendants.
This creates a trade-off: whether to maximize the current 
value achievable from the opportunities or to increase the value achievable
in future generations.
We have shown how to resolve this trade-off by solving a continuous
optimization problem over multiple generations, and we have seen
how optimal solutions to this problem can exhibit a form of
socioeconomic affirmative action, 
favoring individuals
of low socioeconomic status over slightly higher-performing individuals
of high socioeconomic status.
Characterizing the conditions under which this type of 
affirmative action occurs in the model provides insights into the
interaction between the amount of opportunity available,
the magnitude of the gap between different socioeconomic classes, 
and the ``patience''' of the society in evaluating the benefits
achievable from future generations.

\hl{
\xhdr{Insights and Implications}
Our work provides a purely utilitarian account of a society--with no intrinsic interest in reducing inequality-- which nonetheless chooses to employ affirmative action along socioeconomic lines. In that sense, our work responds to concerns around affirmative action hurting social utility. In addition, our analysis present new insights on the shape and extent of effective socioeconomic affirmative action policies across multiple generations. Our findings offer several important insights to the Fair-ML literature: (a) If temporal dynamics are taken into account, there are important cases where fairness interventions can be fully aligned with performance-maximizing goals. (b) Effective fairness interventions should often adapt to population changes over time. (c) For a comprehensive assessment of fairness for an allocation policy/algorithm, we may need to look beyond its immediate impact on its direct subjects. Our work, for instance, characterizes optimal allocations when decisions made for individuals today impact their future generations.
}

\hl{
\xhdr{Toward a Broader Class of Models Leading to Affirmative Action} 
Our work's main focus was on \emph{intergenerational dynamics} and their impact on the effectiveness of \emph{socioeconomic affirmative action policies}. But we hope that our work also serves as a stepping stone toward a broader class of models characterizing conditions under which affirmative action comes out of optimal policy derivations. In particular, a generalized version of our decision process can capture \emph{race-dependent} dynamics of movement between various socioeconomic levels. Race-dependent dynamics are motivated by the observation that advantaged Black people are more likely to be downwardly mobile than their advantaged white counterparts~\citep{chetty2020race}. A more general version of our decision process would consist of race-dependent $\sigma$'s and $\tau$'s. Analyzing the resulting dynamics can shed light on the tradeoffs between race/ethnicity-based and socioeconomic affirmative action policies. We leave this analysis as a crucial direction for future work.
}

There are a number of interesting further research directions 
suggested by the results in the paper, and several of them address the limits of the model noted in the introduction.
In particular, it would be interesting to consider extensions
of the model in which basic parameters of the society varied over
time rather than being fixed and could only be estimated with delayed feedback (e.g., after several years); the ability distributions were non-uniform; the dynamics would amplify the interaction between ability and socioeconomic advantage; or there were multiple
opportunities being allocated concurrently.
It would be interesting to incorporate strategic considerations as
in other work on intergenerational mobility.
And finally, there is a natural opportunity to try integrating
the models here with work on the nature of statistical discrimination
and the design of interventions to alleviate it.

\hl{
\xhdr{On the Scope and Limitations of Mathematical Models}
Our work addresses one potential role that affirmative action policies can have on future generations and mitigating socioeconomic inequality. Our work follows a long tradition in the mathematical social sciences, where a stylized model is proposed to capture certain aspects of a complex societal issue; the model is then rigorously analyzed with the hope that the formal analysis provides new insights about alternative policy choices and their counterfactual effects. These insights can in turn inform policymakers at a qualitative level. 
We conclude this article by acknowledging that all mathematical models are by definition highly simplified representations of the phenomena at hand, and as such it is important to understand and interpret them keeping their limitations and scope of applicability in mind. We have grounded our work in a broad literature from economics so as to draw on the insights that earlier modelers have brought to this setting. But we emphasize that theoretical models and their implications should never be taken as exact representations of the way complex societal processes operate and evolve. As such, it is important to not draw policy interpretations on the basis of such models alone. 
}

\section*{Acknowledgments}
This work was supported in part by a Simons Investigator Award,
a Vannevar Bush Faculty Fellowship,
a MURI grant, AFOSR grant FA9550-19-1-0183, and
grants from the ARO and the MacArthur Foundation.
We are grateful to Lawrence E. Blume, Alexandra Chouldechova, Zachary C. Lipton, Rediet Abebe, Manish Raghavan, Kate Donahue, the AI, Policy, and Practice (AIPP) group at Cornell, and the FEAT reading group at Carnegie Mellon for invaluable discussions. Finally, we would like to thank the anonymous reviewers of our work for their insightful and constructive feedback.

\bibliographystyle{named}
\bibliography{ig_biblio}

\pagebreak
\appendix

\section{Expanded Related Work}\label{sec:related}

\subsection{Long-term Implications of Fair-ML}
Several recent papers study the long-term impact of decision-making models and fairness interventions on society and individuals. \citet{liu2018delayed} and \citet{kannan2018downstream} study how a utility-maximizing \emph{decision-maker} may respond to the predictions made by the model. For instance, the decision-maker may interpret and use the predictions in a certain way, or update the model entirely. 
%
\cite{dong2018strategic,hu2018disparate,milli2018social} address \emph{strategic classification}---a setting in which decision subjects are assumed to respond \emph{strategically} and potentially \emph{untruthfully} to the choice of the classification model, and the goal is to design classifiers that are robust to strategic manipulation. 

\citep{hu2018short} and \citep{mouzannar2019fair} take inspiration from existing models of statistical discrimination and affirmative action. \citet{hu2018short} study the impact of enforcing statistical parity on hiring decisions made in a temporary labor market that precedes a permanent labor market. They show that under certain conditions, statistical parity can result in an equilibrium that Pareto-dominates the one that would emerge in an unconstrained labor market.

\citet{mouzannar2019fair} model the dynamics of how a population reacts to a selection rule by changing their qualifications. The qualification of an individual is assumed to be estimated via some classification function and this estimate is assumed to be the same as their true qualification/label of the individual. They model a selection rule by its selection rates across the qualified and unqualified members of two socially salient groups. In other words, the rule is modeled by four non-negative numbers. They assume that there exist continuously differentiable functions $f_1$ and $f_2$ that map the selection rates in each group to the percentage of qualified individuals in that group in the next round. These functions model how qualifications change over time in response to the selection rule. Within this model,  \citet{mouzannar2019fair} study two types of myopic utility-maximizing policies: an affirmative action policy which forces the selection rates to be equal across the two groups, and an unconstrained policy that simply picks the qualified in each group. They then provide several conditions under which each of these policies lead to social equality (i.e., groups that are indistinguishable in terms of their qualifications). As for affirmative action policies, under-acceptance of qualified individuals in one group (to guarantee equal selection rates) is shown to be inefficient (both in terms of the decision-maker's utility and average qualifications). Over-acceptance of the unqualified, on the other hand, can lead to a more qualified population (although it may fail to guarantee social equality).

\subsection{Intergenerational Income Mobility}
A large economic literature addresses the relationship between income inequality and intergenerational mobility. At a high-level, the literature suggests that higher inequality results in lower income mobility across generations (see e.g.,~\citep{maoz1999intergenerational,corak2013income}). The precise measurements of inequality and mobility significantly influence the strength of this effect (see, e.g.,~\citep{solon1992intergenerational,piketty2000theories}). The literature recognizes two main mechanisms through which a parent may affect its offspring income: (1) by transmitting endowments (e.g., innate abilities, connections, etc.) to the next generation; (2) by investing in the human capital of the next generation~\citep{becker1979equilibrium,solon1999intergenerational}. 
In the existing theoretical models of income mobility, utility-maximizing parents decide how much of their income/human capital to consume and how much of it to invest in their offspring (depending on their level of altruism)~\citep{becker1986human,becker1979equilibrium,loury1981intergenerational,solon1999intergenerational}. This decision, combined with the social institutions and the technology that converts parental investment to the next generation's human capital, determine the fate of the next generation.

For example, \citet{loury1981intergenerational} provides a dynamic model of how the earnings distributions of the next generation depends on the parent's earnings and his/her investment decision for his/her offspring. The earnings of the offspring is determined by their innate ability/endowment $\alpha$\footnote{Similar to ours,  \citeauthor{loury1981intergenerational}'s model the ability of each individual is drawn randomly from some fixed distribution. The ability of the offspring is only observed after the parent makes their investment decision.} and how much the parent invests in their training $e$ through a function $h(\alpha, e)$. Unlike our model, \citeauthor{loury1981intergenerational} assumes the the parents are expected-utility maximizing when making their investment decisions for their offspring. In particular, he assumes the utility $u(.,.)$ of a parent with earnings $y$ is a function of his/her own consumption $c$ and the expected utility of the offspring as the result of the investment he/she makes in their training, which is $y-c$. This recursive definition of utility specifies how the utility of one generation depends on that of all their subsequent generations (not just the immediate offspring). \citeauthor{loury1981intergenerational} defines the indirect utility associated with an earning $y$ to be the function the parent uses to estimate their offspring's utility. If this function is the same as the one obtained by solving the parent's utility maximization problem, the indirect utility function is called \emph{consistent}. He goes on to show that under certain (weak) assumptions, there is a unique consistent indirect utility function. He also defines an \emph{equilibrium} earnings distribution as one that if it characterizes the earnings distribution for one generation, it continues to do so for all subsequent generations. He then provides several comparative statics results for various policies (e.g., ``education specific tax policies''). For example, he shows that under certain conditions, egalitarian policies that redistribute earnings of the next generation have insurance effects that make every member of society today better off.

%

Other factors that have been shown to causally affect income mobility are neighborhood~\citep{chetty2018impacts,chetty2018impacts},
parental education~\citep{torche2011college},
and family background characteristics (e.g., welfare receipt, health, attitudes and social behavior~\citep{black2010recent}).

\subsection{Affirmative Action Policies}
A rich body of work in economics investigates sources of statistical discrimination\footnote{According to \citep{fang2011theories} ``Statistical discrimination generally refers to the phenomenon of a decision-maker using observable characteristics of individuals as a proxy for unobservable, but outcome-relevant, characteristics.''} and the role of affirmative action policies in redressing it. For an excellent survey of this literature, see~\citep{fang2011theories}.

In contrast to taste-based theories of discrimination---which attribute group inequality to racial or gender preference against members of certain groups---statistical discrimination theories cast group inequality as a consequence of interaction between two \emph{rational} parties: 
\begin{enumerate}
\item A utility maximizing decision-maker (e.g., an employer) who has imperfect information about a decision subject's characteristics and uses his/her group membership as a signal for his/her outcome-relevant unobservable characteristics (e.g., employee's productivity);
\item Several groups of individuals (e.g., racial or gender groups) that are ex-ante identical in terms of the distribution of qualifications. Individuals are utility maximizing and best-respond to the decision maker's strategy by adjusting their investment in qualifications.
\end{enumerate}
In a seminal work, \citet{arrow1973theory} argues that differences between groups can be explained as a form of \emph{coordination failure}: In equilibrium, the decision maker holds asymmetric beliefs about group qualifications, and this serves as a \emph{self-fulfilling stereotype}. Because of this belief, members of the disadvantaged group don't have enough incentive to invest in skills and qualifications---precisely because they know that the decision maker will treat them unfavorably). This in turn rationalizes the decision maker's belief---members of the disadvantaged group indeed end up being less qualified (on average) than the advantage group. 

Outcome-based policies, such as affirmative action quotas or the application of disparate impact tests, have long been proposed and implemented as a \emph{temporary} remedy to eliminate group-level inequalities. Such policies may seem particularly effective when self-fulfilling stereotypes are the primary cause of group inequalities.  Imposing quota constraints can lead the players to coordinate on a symmetric outcome, that is, and equilibrium in which the decision maker holds symmetric beliefs about different groups. This, however, is not the only possible consequence of imposing quota constraints. \citet{coate1993will} have shown that quota constraints may reduce the disadvantaged group's incentives to invest in skills, and subsequently, result in an even worse asymmetric the equilibrium. \citeauthor{coate1993will} call this phenomenon \emph{patronization}. (This is a potential consequence of algorithmic fairness-enhancing interventions.)

Advocates of affirmative action have often argued that the larger representation of minorities in higher social positions can generate \emph{role models} that can positively influence future generations of minorities in their investment decisions. \citet{chung2000role} formalizes these arguments, by allowing for groups to differ in their costs of investment.

While the existing literature on affirmative action largely focuses on \textit{race-based} interventions and policies, some scholars have advocated for \emph{socio-economic} or \emph{class-based} affirmative action as a race-neutral alternative (see e.g., \citep{kahlenberg1996class,carnevale2013socioeconomic}). Race-neutral alternatives become particularly important to understand when race-sensitive interventions are constitutionally challenged (for example, see the case of Fisher v. University of Texas at Austin (2013)). 
The extent to which substituting socio-economic status for race in
college admissions is effective in improving minority enrollment is contested (see, e.g., ~\citep{kane1998racial,reardon2006implications,gaertner2013considering,reardon2017can}). However, socio-economic affirmative action can certainly facilitate access to college for economically disadvantaged students. The resulting socioeconomic diversity is a desirable outcome in and of itself~\citep{reardon2017can} and in essence, it does not have to compete with or replace race-sensitive admission policies.

{
\subsection{Comparison with \citep{durlauf2008affirmative}}\label{app:durlauf}
\citep{durlauf2008affirmative} is one of the closest papers to our work---in terms of the motivation and modeling choices. In his concluding remarks, Durlauf poses the key question our work sets out to answer: how do multigenerational considerations impact the optimal allocation policy? Next, we provide a brief summary of \citep{durlauf2008affirmative} and compare our work with it in terms of modeling choices and results.

\citeauthor{durlauf2008affirmative} focuses on a condition under which efficiency and equality considerations are aligned, and that is when a diverse body of students on college campuses improves the human capital development for all of them. 
\citet{durlauf2008affirmative} provides a model to compare the equality and efficiency of affirmative action policies with those of meritocratic policies in the context of admission rules to public universities. An admission rule maps the initial human capital of students (and possibly their demographic information, such as race) to their admission decisions. If a student is admitted, their human capital is then improved by the quality of the college they attend and the average human capital of their classmates in that college. The state wishes to employ an admission policy that maximizes the aggregate human capital. He argues that while the meritocratic policy may be efficient under certain conditions, the same holds for affirmative action policies (e.g., when ``diversity improves education for all students.'').  

In many respects, our model is similar to \citep{durlauf2008affirmative}, although we depart from his model and analysis in a few key areas. Most importantly, we model the dynamics of human-capital development across \emph{multiple generations} and derive the \emph{optimal allocation policy}.

Similar to our work, Durlauf studies a multi-generational model of a society, wherein each generation, a new member is born into every family/dynasty, and he/she replaces the current adult member of the family in the next generation. Even though our model is not expressed in terms of \emph{overlapping generations}, one can equivalently cast it in those terms.

The \emph{initial human capital} in Durlauf's model corresponds to our notion of \emph{success probability}. \emph{Adult human capital} in his model is determined by college attendance, and it roughly maps to our notion of \emph{success} (i.e., whether the individual succeeds if given the opportunity.)

In Durlauf's model, students always prefer classmates with higher human capital and colleges with higher fundamental quality. Unlike his model, we do not capture the \emph{spillover effects} of students attending the same college on each other. We also don't allow for various levels of college quality. (It may be worth noting that while in his general model, Durlauf allows for various college qualities, in his analysis he focuses on at most two quality levels---one college of high quality, the other of low quality.)

In both models, an admission rule maps a student's human capital and group membership into a binary outcome indicating whether the student is admitted. In our work, groups corresponds to \emph{socio-economic tiers}, whereas in Durlauf's, they correspond to \emph{racial groups}. In both models, the admission rule may vary across time and generations.

Other simplifying assumptions in Durlaufs model are:
\begin{itemize}
\item The level of state expenditures on education is fixed across generations. Similarly, alpha is fixed in our model.
\item The only output of universities is human capital. Similarly, we only care about the percentage of the population who succeed.
\item The initial human capital is accurately measured for every student. Similarly, we assume ability and success probability are perfectly observable.
\end{itemize}

As for the objective function, Durlauf assumes the policymaker wants to maximize the level of human capital among \emph{the next generation} of adult citizens. He restricts attention to two generations only, but he notes that ``When one moves to this dynamic perspective, one also needs to consider an \emph{inter-temporal} objective function; in parallel to my simplest specification, which is based on the level of human capital at one point in time, it is appropriate to consider a \emph{weighted average} of the levels across time. A standard way to make this move involves working with [...] a \emph{discounted sum} of adult human capital across two generations.'' The latter is precisely the objective function we aim to optimize. Unlike our work, Durlauf does not solve for the optimal policy. He merely compares the meritocratic policy with affirmative action in terms of the average human capital of the next generation. (Similar to our work, he defines a meritocratic admissions policy to mean that each student is admitted to college exclusively on the basis of their initial human capital. An affirmative action policy in his view is one that takes group membership into account as well.)
 
Both models take the perspective of a policymaker who needs to choose between affirmative action and meritocratic rules and aim to understand how efficiency and equity interact. Similar to our findings, he finds that depending on the parameters of the model, affirmative action can be more efficient than myopic meritocratic admission rules. In our model, this happens because the human capital of a student is directly impacted by that of his/her parent. In Durlauf's work, in contrast, this can happen because of the spillover effects of students on each other and the impact of college quality on developing adult human capital. In more concrete terms, Durlauf's model can capture the competing claims that "diversity improves education for all students on campus", or "stronger students going to the same college leads to higher social capital for all of them". Our model does not capture such \emph{spillover} effects. We instead focus on the \emph{intertemporal} or intergenerational effects of admission rules. As mentioned earlier, Durlauf emphasizes the importance of intergenerational factors in Section 6.2 of "Affirmative action, meritocracy, and efficiency". 
}

\subsection{MDPs with Continuous State and Action Spaces}\label{app:continuous}
Finding the optimal policy for a continuous state- and action-space MDP is often a difficult task~\citep{marecki2006fast}.
Two main categories of algorithms have
been proposed to (approximately) calculate the optimal policy: 
(1) A typical approach is to \textbf{discretize} the state and action space, then solve the resulting MDP using standard methods, such policy or value iteration. This approach suffers from the well-known ``curse of dimensionality''.
(2) Another approach \textbf{approximates} the optimal value function with a parametric form, then sets out to fit those parameters such that they satisfy the Bellman equation (see, e.g., ~\citep{li2005lazy,hauskrecht2004linear}.

\section{Limitations and Interpretations}\label{sec:limitations}
Our model is designed to incorporate the basic points we mentioned in the introduction in as simplified a fashion as possible; as such, it is important
to note some of its key limitations.
First, it is intended to model the effect of a single opportunity, and it treats other forms of mobility probabilistically in the background.
It also assumes that its fundamental parameters ($\alpha, \sigma, \tau, \gamma$) as constant over all generations.
It treats an individual's group membership ($A$ and $D$) and
ability as a complete description of their performance, rather than including any dependence on the group membership of the individual's parent. (That is, an individual in group $A$ performs the same in the model regardless of whether their parent belonged to group $A$ or $D$.)
All of these would be interesting restrictions to relax in an
extension of the model.

Much of the past theoretical work on intergenerational mobility focuses
on an issue that we do not consider here: the strategic considerations
faced by parents as they decide how much to consume in the present
generation and how much to pass on to their children.
Our interest instead has been in the optimization problem faced by
a social planner in allocating opportunities, treating the behavior of the agents as fixed and simple. Here too, it would be interesting to explore models that address these issues in combination.

Finally, because our focus is on intergenerational mobility in
a socioeconomic sense, we do not model discrimination based on race, ethnicity, or gender, and the role of race-based and gender-based affirmative action in combatting these effects.
The model is instead concerned with
\emph{socio-economic} or \emph{class-based}~\citep{malamud1995class,kahlenberg1996class} affirmative action.
That said, the ingredients here could be combined with models of statistical or taste-based discrimination on these attributes to better understand their interaction.

The simplicity of our model, however, does allow us to make a correspondingly fundamental point: that even a purely payoff-maximizing society can discover affirmative action policies from first principles as
it seeks to optimize the allocation of opportunities over multiple 
generations.
Moreover, the optimal allocation policy is deeply connected
to dynamic programming over the generations;
the society is essentially attempting to ``steer'' the balance
of group $A$ and group $D$ over time, making sure not to turn things
too abruptly (giving up present benefit) or too gradually (giving up
future benefit).

\section{Properties of the Value Function}\label{app:value}
In this section, we show that the value function $V(.)$ for the decision process $\cD$ is unique, continuous, differentiable, and monotonically decreasing. We begin by offering an alternative formulation of the decision process that has the exact same optimal policy and value function, but is more conducive to recursive analysis. Next, we prove several properties of the state space $\Phi$ and the reward function $R'$ for this alternative decision process. Then we apply previously-established theorems from dynamic programming and recursive analysis~\citep{stokey1989recursive} to establish the aforementioned properties of $V(.)$.

\paragraph{A Decision Process Equivalent to $\cD$}
Given the decision process $\cD =(\Phi, \Theta, S, R)$ (as defined in Section~\ref{sec:DP}), we first provide an alternative formulation, called $\cD'$, that fits the standard representation of decision processes in the Dynamic Programming literature. This standard formulation allows us to import tools and theorems from Dynamic Programming with minor modifications. 

We construct $\cD'$ from $\cD$ by re-defining the action space---not in terms of thresholds, but in terms of the states reachable through an admissible threshold from a given state. The state transitions in this formulation become trivial, but we need to re-write the reward function in terms of the new parameters (i.e., current and the next state). 

Formally, given $\cD$, we define $\cD' =(\Phi, \Gamma, I, R')$ as follows: 
\begin{itemize}
\item The correspondence $\Gamma: \Phi \rightarrow \Phi$ specifies the states reachable by one admissible $\theta_0$ from any given state $\phi_0 \in \Phi$. More precisely
$$\Gamma(\phi_0)=\{\omega \in \Phi \vert \exists \theta_0 \in \Theta(\phi_0) \text{ s.t. }\omega = S(\phi_0, \theta_0)\}.$$
\item $I: \Phi_0 \times \Phi_0 \rightarrow \Phi_0$ simply returns its second argument, that is, for all $\phi_0 \in [0,1]$, $I(\phi_0,\omega) = \omega$.
\item To recast the reward function in terms of $(\phi_0,\omega)$, we first write $\theta_0$ as a function of $\phi_0, \omega$:
{\scriptsize
\begin{eqnarray*}
&& \omega = \phi_0 - \frac{\phi_0}{2\sigma}(\sigma^2 - \theta^2_0) \\
&\Leftrightarrow & 2\sigma \frac{\omega}{\phi_0} + \sigma^2 - 2\sigma = \theta^2_0\\
&\Leftrightarrow & \theta_0 = \sqrt{2\sigma\frac{\omega}{\phi_0} + \sigma^2 - 2\sigma}
\end{eqnarray*}
}
Given the above change of variables, we can write:
{\scriptsize
$$R'(\phi_0, \omega_0) = 
\begin{cases}
R \left(\phi_0, \sqrt{2\sigma\frac{\omega}{\phi_0} + \sigma^2 - 2\sigma} \right) \quad \text{ if } \phi_0 > 0\\
R(0,\sigma) \quad \text{ if } \phi_0 = 0
\end{cases}
$$
}
\end{itemize}

\begin{property}\label{4.3}
$\Phi$ is a convex subset of $\reals$, and the correspondence $\Gamma$ is nonempty, compact-valued, and continuous.
\end{property}
\begin{proof}
$\Phi = [0,1]$, which is clearly a convex subset of $\reals$. The correspondence $\Gamma$ can be characterized as follows:
{\tiny
\begin{eqnarray*}
&& \Gamma(\phi_0) \\
&=&\{\omega \in \Phi \vert \exists \theta_0 \in \Theta(\phi_0) \text{ s.t. }\omega = S(\phi_0, \theta_0)\} \\
&=& \left\{ \omega \in \Phi \vert \omega = \phi_0 - \frac{\phi_0}{2\sigma} (\sigma^2 - \theta_0^2) \text{ where } \sigma \left( 1 - \frac{\alpha}{\phi_0} \right) \leq \theta_0 \leq \frac{\sigma(1-\alpha)}{\phi_0} \text{ and } 0 \leq \theta_0 \leq 1 \right\} \\
&=& \begin{cases}
               \omega \in \left[\phi_0\left(1-\frac{1}{2\sigma}\right) , \phi_0\right] \quad \text{ if } \phi_0 \leq \alpha\\
               \omega \in \left[\phi_0-\sigma\alpha \left(1+\frac{\alpha}{2\phi_0}\right) , \phi_0\right] \quad \text{ if } \alpha < \phi_0 \leq 1-\alpha\\
               \omega \in \left[\phi_0-\sigma\alpha \left(1+\frac{\alpha}{2\phi_0}\right)  , \phi_0-\frac{\phi_0}{2\sigma} \left( 1 - \frac{1-\alpha}{\phi_0} \right)\left( 1 + \frac{1-\alpha}{\phi_0} \right)\right] \quad \text{ if } \phi_0 \geq 1-\alpha
            \end{cases}
            \\
&=& \left[ \max \left\{ \phi_0\left(1-\frac{1}{2\sigma}\right), \phi_0-\sigma\alpha \left(1+\frac{\alpha}{2\phi_0}\right) \right\}  , \min\left\{ \phi_0, \phi_0-\frac{\phi_0}{2\sigma} \left( 1 - \frac{1-\alpha}{\phi_0} \right)\left( 1 + \frac{1-\alpha}{\phi_0} \right) \right\}\right]            
\end{eqnarray*}
}
Given the above definition, it is trivial to verify that $\Gamma$ is indeed nonempty, compact-valued, and continuous.
\end{proof}

\begin{property}\label{4.4}
The reward function $R'$ is bounded and continuous.
\end{property}
\begin{proof}
The reward function $R$ specifies the fraction of the population who succeed if given the opportunity. So clearly $0 \leq R(.,.) \leq 1$ is bounded. As a result $R'$ is also bounded. 
To establish continuity of $R'$, we first establish the continuity of $R$. It is trivial to see that $R(.,.)$ (defined in (\ref{eq:reward})) is continuous at any $\phi_0<1$. At $\phi_0 = 1$, we have

{\tiny
\begin{eqnarray*}
&& \lim_{\phi_0 \rightarrow 1} \frac{\phi_0}{2\sigma} (\sigma^2 - \theta_0^2)
+ \frac{1-\phi_0}{2\sigma} \left((\sigma+\tau)^2 - \left( \frac{\sigma(1-\alpha) + (1-\phi_0)\tau - \phi_0 \theta_0}{1-\phi_0}\right)^2\right)\\
&=& \frac{1}{2\sigma} (\sigma^2 - \theta_0^2) - \lim_{\phi_0 \rightarrow 1} 
\frac{1-\phi_0}{2\sigma} \left( \frac{\sigma(1-\alpha) + (1-\phi_0)\tau - \phi_0 \theta_0}{1-\phi_0}\right)^2\\
&=& \frac{1}{2\sigma} (\sigma^2 - \theta_0^2) - \lim_{\phi_0 \rightarrow 1} 
\frac{1}{2\sigma}  \frac{\left(\sigma(1-\alpha) + (1-\phi_0)\tau - \phi_0 \theta_0\right)^2}{1-\phi_0}\\
\text{($\Gamma(1) = \{\sigma(1-\alpha)\}$) } &=& \frac{1}{2\sigma} (\sigma^2 - \theta_0^2) - \lim_{\phi_0 \rightarrow 1} 
\frac{1}{2\sigma}  \frac{\left(\sigma(1-\alpha) + (1-\phi_0)\tau - \phi_0 \sigma(1-\alpha)\right)^2}{1-\phi_0}\\
\text{(L'Hospital's rule) }&=& \frac{1}{2\sigma} (\sigma^2 - \theta_0^2) - \lim_{\phi_0 \rightarrow 1} 
\frac{1}{\sigma}  \frac{ (\tau+\sigma(1-\alpha))\left(\sigma(1-\alpha) + (1-\phi_0)\tau - \phi_0 \sigma(1-\alpha)\right)}{1}\\
&=& \frac{1}{2\sigma} (\sigma^2 - \theta_0^2) - 0\\
 &=& \frac{1}{2\sigma} (\sigma^2 - \theta_0^2) \\
 &=& R(1, \theta_0).
\end{eqnarray*}
}
So $R$ is continuous at $\phi_0 = 1$, as well.
Finally, note that $\sqrt{2\sigma\frac{\omega}{\phi_0} + \sigma^2 - 2\sigma}$ is continuous function of $(\phi_0, \omega) \in \Gamma(\phi_0)$, so $R'$ is also continuous. 
\end{proof}

Let $C(\Phi)$ be the space of bounded continuous functions $f:\Phi_0 \rightarrow \reals$ with the sup norm $\Vert f \Vert = \max_{\phi_0 \in \Phi} \vert f(\phi_0)\vert$.
We define the operator $T$ on the space $C(\Phi)$ as follows:
$$(Tf)(\phi_0) = \max_{\omega \in \Gamma(\phi_0)} R'(\phi_0,\omega) + \gamma f(\omega).$$

\begin{theorem}[Adapted version of Theorem 4.6 in \citep{stokey1989recursive}] Let $\Phi$, $\Gamma$, and $R'$ satisfy Properties~\ref{4.3} and \ref{4.4}. Then the operator $T$ maps $C(\Phi)$ into itself, $T:C(\Phi) \rightarrow C(\Phi)$, and $T$ has a unique fixed point $V \in C(\Phi)$. Moreover, given $V$, the optimal policy correspondence $\Pi^*: \Phi \rightarrow \Phi$ is compact-valued and upper hemi continuous (u.h.c).
\end{theorem}

According to \citet{cotter2006non}, if the reward function $R'$ is differentiable, the value function $V$ is differentiable on any interior point that is an optimal ``next state'' for some current state. More precisely,

\begin{theorem}[Adapted version of Theorem 2 in \citep{cotter2006non}]
Suppose $\omega \in \Pi^*(\phi_0) \cap (0,1)$ for some $\phi_0 \in [0,1]$. If $R'$ is continuously differentiable, then $V$ is differentiable at $\omega$, with $\dd{\phi_0}V \vert_{\omega}=\dd{\phi_0}R' \vert_{(\omega, \omega')}$ for any $\omega' \in \Pi^*(\omega)$.
\end{theorem}

It only remains to show that:
\begin{property}\label{4.9}
$R'$ is continuously differentiable on the interior of its domain.
\end{property}
\begin{proof}
To establish continuous differentiablity of $R'$ with respect to $\omega$, note that:
{\scriptsize
$$\dd{\omega} R'(\phi_0, \omega) = \dd{\theta_0} R(\phi_0, \theta_0) \dd{\omega} \theta_0.$$
}
and both terms on the right hand side of the above equation are continuous.
{\scriptsize
$$\dd{\theta_0} R(\phi_0, \theta_0) = \frac{\phi_0}{\sigma} \left( \frac{\sigma(1-\alpha) + (1-\phi_0)\tau - \theta_0}{1-\phi_0} \right) $$
$$\dd{\omega} \theta_0  = \frac{\sigma}{\sqrt{2\sigma\omega\phi_0 + (\sigma^2 - 2\sigma)\phi_0^2}}$$
}
Therefore, $\dd{\omega} R'(\phi_0, \omega)$ is trivially continuous at any $(\phi_0,\omega) \in (0,1)^2$.

To establish continuous differentiability of $R'$ with respect to $\phi_0$, note that:
{\scriptsize
$$\dd{\phi_0} R'(\phi_0, \omega) = \dd{\phi_0} R(\phi_0, \theta_0)  + \dd{\theta_0} R(\phi_0, \theta_0) \dd{\phi_0} \theta_0.$$}
It is easy to see that all three terms in the right hand side of the above are continuous:
{\scriptsize
\begin{eqnarray*}
\dd{\phi_0} R(\phi_0, \theta_0)  &=& \frac{1}{2\sigma} (\sigma^2 - \theta_0^2) - \frac{1}{2\sigma} \left((\sigma+\tau)^2 - \left( \frac{\sigma(1-\alpha) + (1-\phi_0)\tau - \phi_0 \theta_0}{1-\phi_0}\right)^2\right)\\
&&- \frac{1}{\sigma}\left( \frac{\sigma(1-\alpha) + (1-\phi_0)\tau - \phi_0 \theta_0}{1-\phi_0}\right) \left( \frac{\sigma(1-\alpha) - \theta_0}{1-\phi_0} \right)
\end{eqnarray*}
}
{\scriptsize
$$\dd{\phi_0} \theta_0  = \frac{\sigma\omega}{\phi_0^2 \sqrt{2\sigma\omega\phi_0 + (\sigma^2 - 2\sigma)\phi_0^2}}$$
}
Therefore, $\dd{\phi_0} R'(\phi_0, \omega)$ is continuous at any $(\phi_0,\omega)$ in the interior of $R'$s domain.
\end{proof}

\begin{proposition}
$V(\phi_0)$ is monotonically decreasing at all $\phi_0 \in (0,1)$.
\end{proposition}
\begin{proof}
For any $\phi_0 \leq \phi^*_0$, the statement is easy to verify given the closed form expression for the value function in (\ref{eq:value}).
For any $\phi_0 >\phi^*_0$ we show that $V$ is decreasing in an open neighborhood on the left side of $\phi_0$. Since $V$ is differentiable at $\phi_0$, this implies $V'(\phi_0) \leq 0$. 

Let $\phi'_0$ be the state we get to if we apply the optimal threshold $\theta_0$ at $\phi_0$. We have that $\phi'_0 < \phi_0$ (note that at $\phi_0 >\phi^*_0$, $\theta_0 < \sigma$, which implies $\phi'_0 < \phi_0$). Now for any state $\phi'_0 < \phi''_0 < \phi_0$, we can show that $V(\phi''_0)>V(\phi_0)$. This is simply because we can reach $\phi'_0$ from $\phi''_0$ using a threshold $\theta''_0 > \theta_0$. To see this, note that:
{\scriptsize
\begin{eqnarray*}
&& S(\phi_0, \theta_0) = \phi'_0  = S(\phi''_0, \theta''_0) \\
&\Leftrightarrow &\phi_0 - \frac{\phi_0}{2\sigma} (\sigma^2 - \theta_0^2) = \phi''_0 - \frac{\phi''_0}{2\sigma} (\sigma^2 - {\theta''_0}^2)\\
&\Leftrightarrow & \phi_0 (1 - \frac{\sigma}{2} + \frac{\theta_0^2}{2\sigma}) = \phi''_0 (1 - \frac{\sigma}{2} + \frac{{\theta''_0}^2}{2\sigma})
\end{eqnarray*}
}
Since $\phi_0 > \phi''_0$, it must be the case that ${\theta''_0} > \theta_0$ for the above equation to hold.

Next, observe that compared to applying $\theta_0$ at $\phi_0$, using the higher threshold of $\theta''_0$ at $\phi''_0$ leads to a higher immediate reward and the same next state value at $\phi''_0$. From this observation, we can conclude that the value of $\phi''_0$ is higher than that of $\phi_0$, because:
{\scriptsize
\begin{eqnarray*}
V(\phi''_0) &\geq & r(\phi''_0, \theta''_0) + \gamma V(\phi'_0) \\
&\geq & r(\phi_0, \theta_0) + \gamma V(\phi'_0) \\
&= & V(\phi_0).
\end{eqnarray*}
}
\end{proof}

\begin{figure*}[h!]
    \centering
    \begin{subfigure}[b]{0.31\textwidth}
        \includegraphics[width=\textwidth]{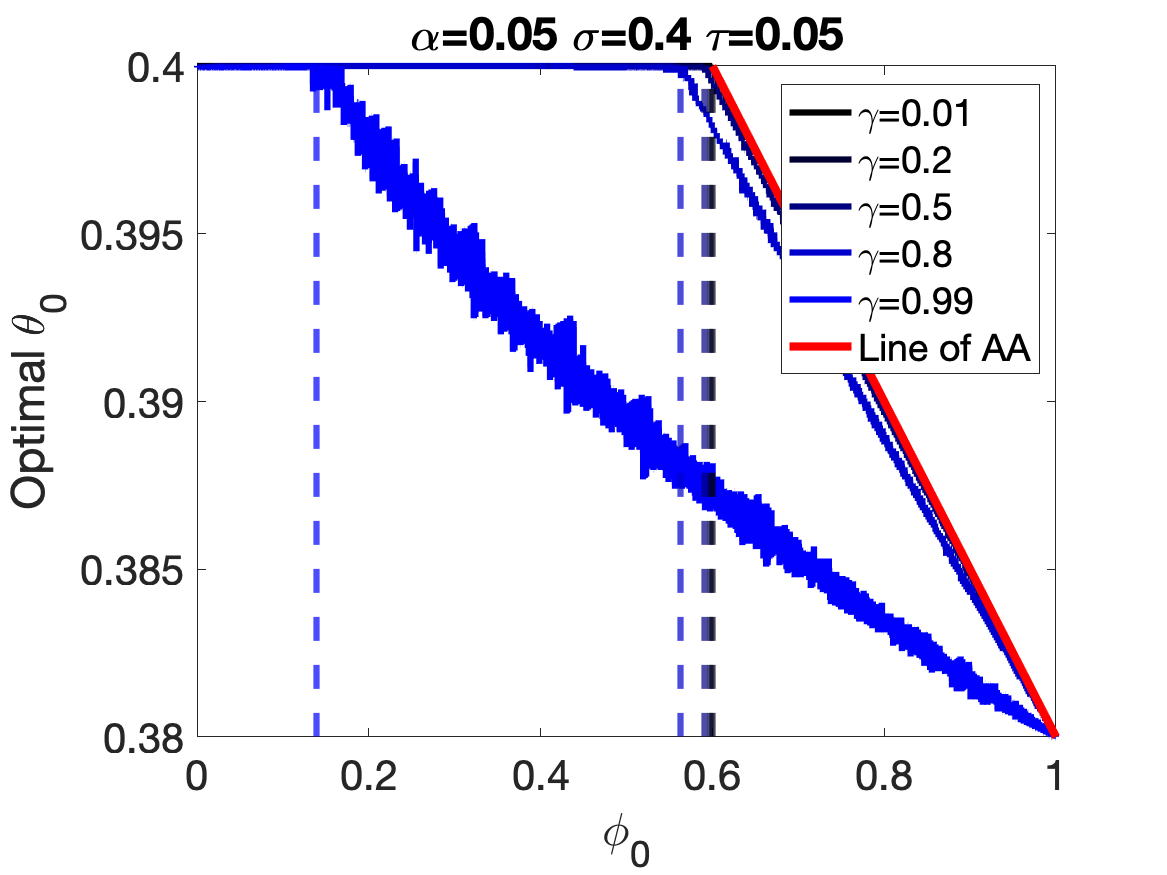}
    \end{subfigure}
    \begin{subfigure}[b]{0.31\textwidth}
        \includegraphics[width=\textwidth]{Figures/Policy_Gamma_alpha_15_sigma_40_tau_10.png}
    \end{subfigure}
    \begin{subfigure}[b]{0.31\textwidth}
        \includegraphics[width=\textwidth]{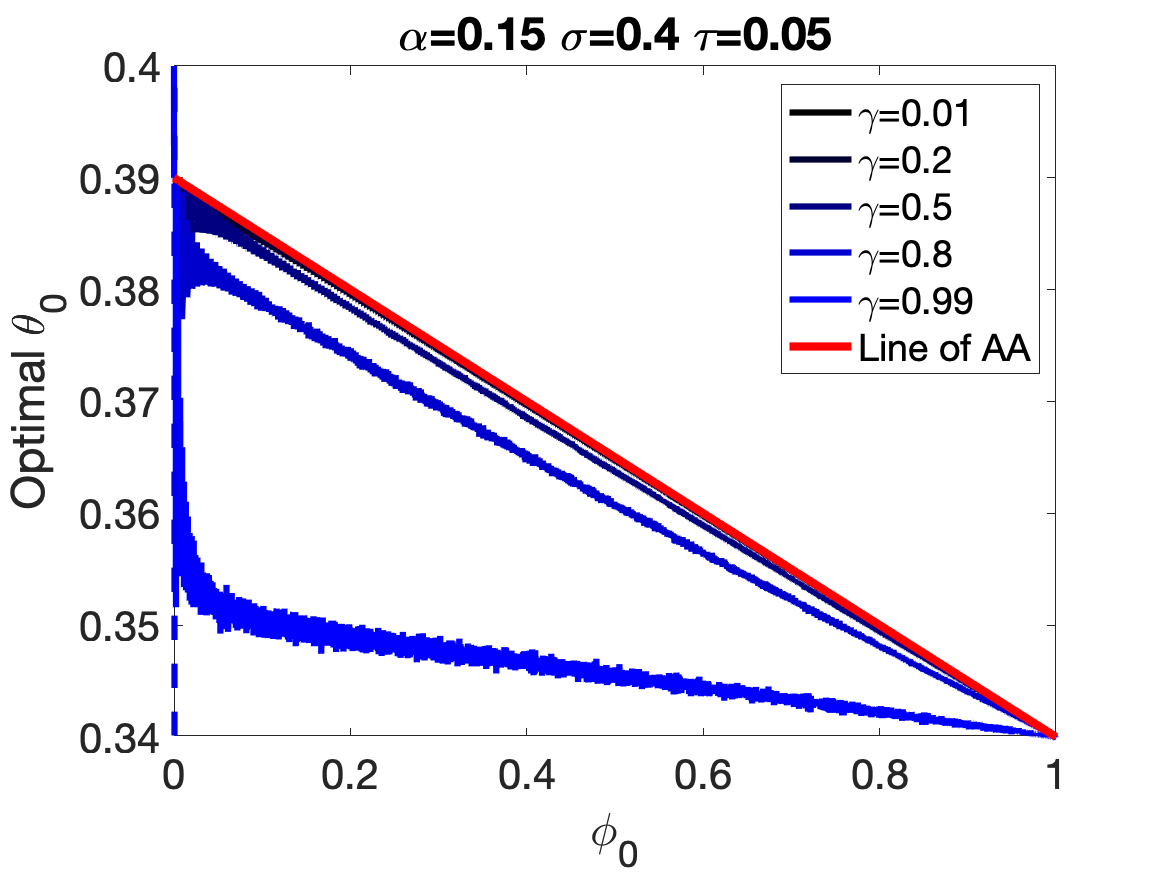}
    \end{subfigure}
   \caption{The optimal policy $\theta_0$ at every state $0 \leq \phi_0 \leq 1$ for various settings of $\alpha, \sigma, \tau$, and $\gamma$. Note that in all cases, \textbf{the optimal threshold is monotonically decreasing with $\phi_0$}. Moreover, there exists a point below which $\sigma$ is the only optimal threshold, and above which $\sigma$ is no longer optimal. Observe that this point coincides with the tipping point, $\phi^*_0$, established in Theorem~\ref{thm:main} (depicted by dashed lines). The dotted red line illustrates the line of affirmative action, derived in Lemma~\ref{lem:AA-line}. Notice that \textbf{beyond $\phi^*_0$, the optimal policy is always below the line of affirmative action}.}\label{fig:policy}
\end{figure*}
\begin{figure*}[h!]
    \centering
    \begin{subfigure}[b]{0.31\textwidth}
        \includegraphics[width=\textwidth]{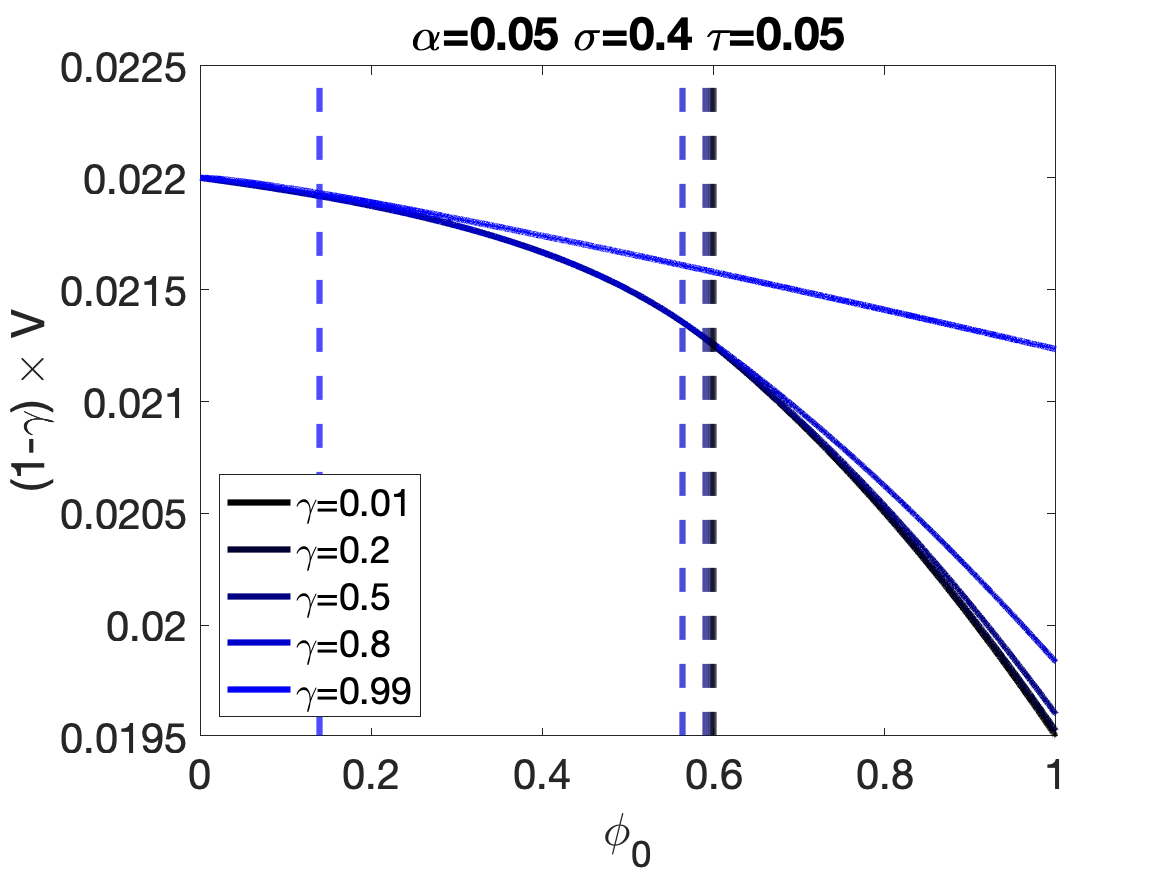}
    \end{subfigure}
    \begin{subfigure}[b]{0.31\textwidth}
        \includegraphics[width=\textwidth]{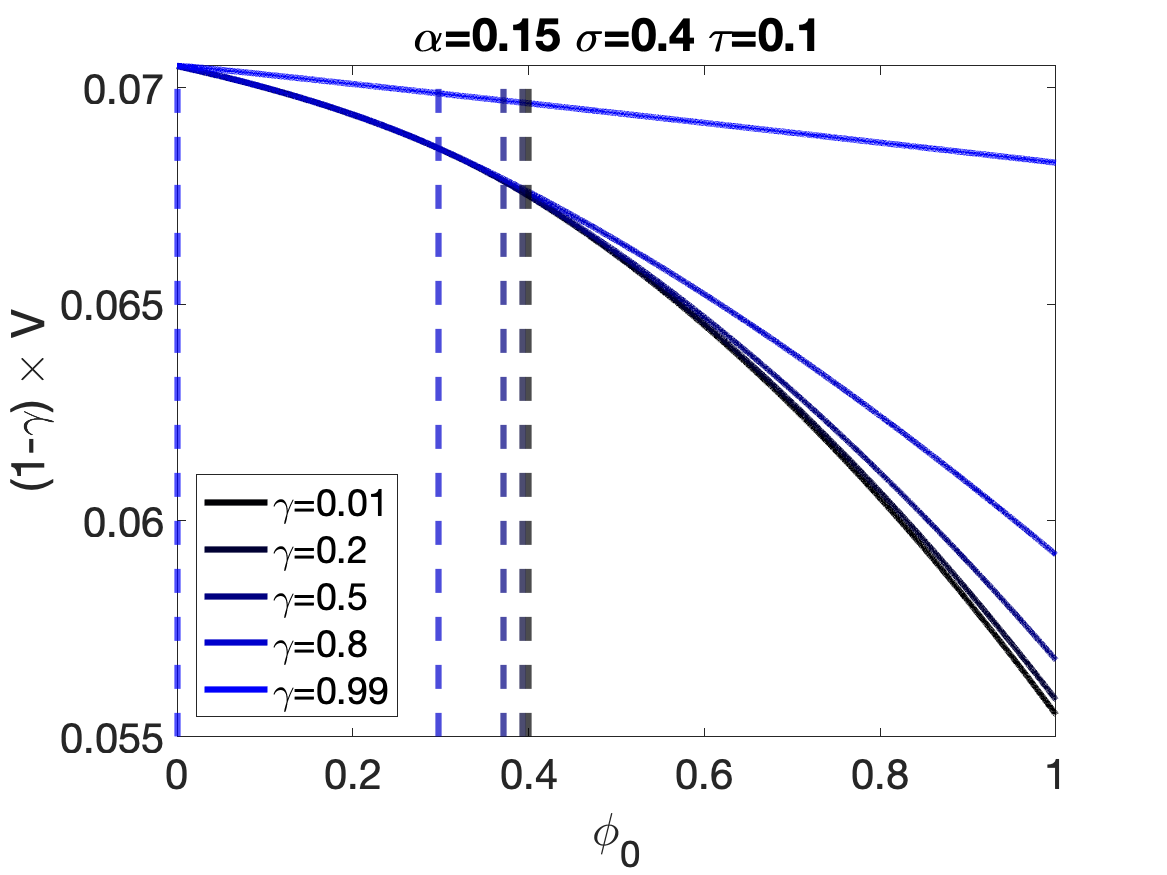}
    \end{subfigure}
    \begin{subfigure}[b]{0.31\textwidth}
        \includegraphics[width=\textwidth]{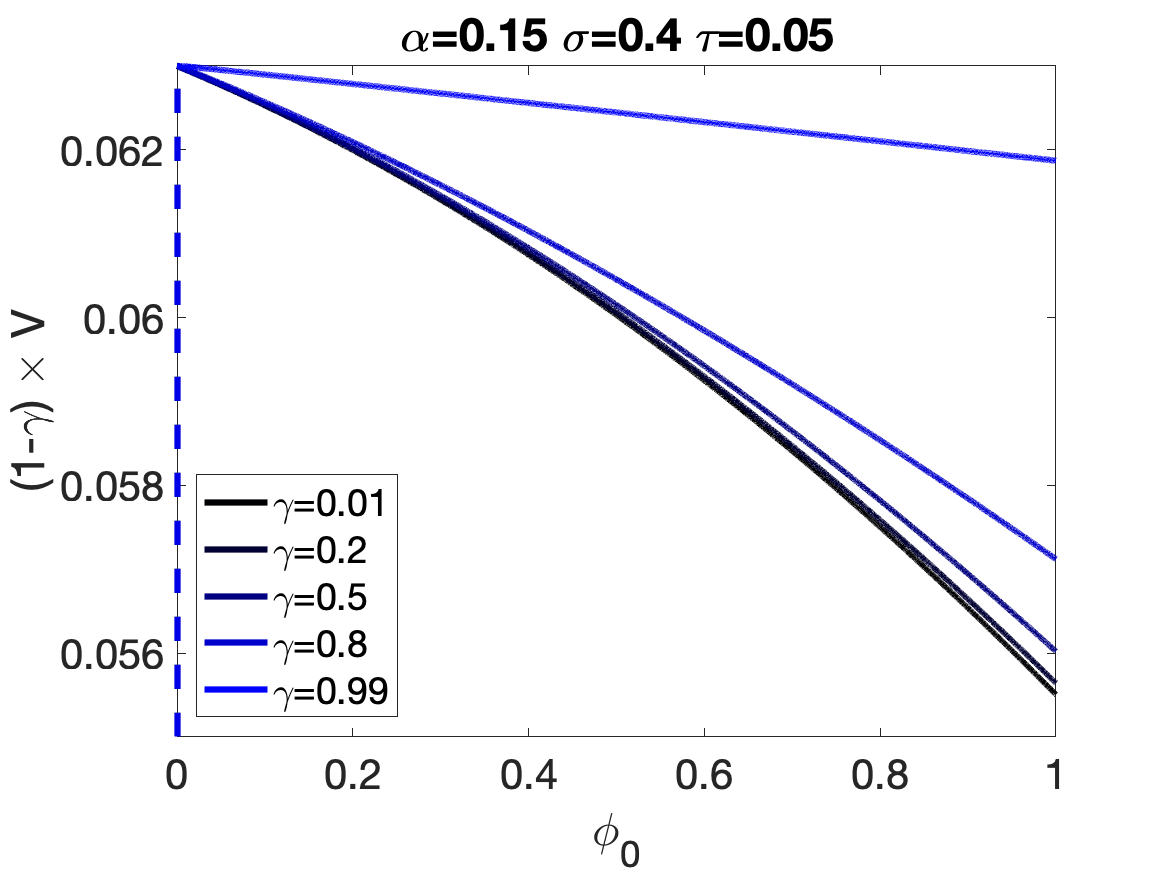}
    \end{subfigure}
   \caption{The scaled value function at every state $0 \leq \phi_0 \leq 1$ for various settings of $\alpha, \sigma, \tau$, and $\gamma$. The dashed lines specify the tipping point, $\phi^*_0$. Note that in all cases, \textbf{the value function is continuous, concave, and decreasing}.}\label{fig:value}
\end{figure*}

\begin{figure*}[h!]
    \centering
    \begin{subfigure}[b]{0.31\textwidth}
        \includegraphics[width=\textwidth]{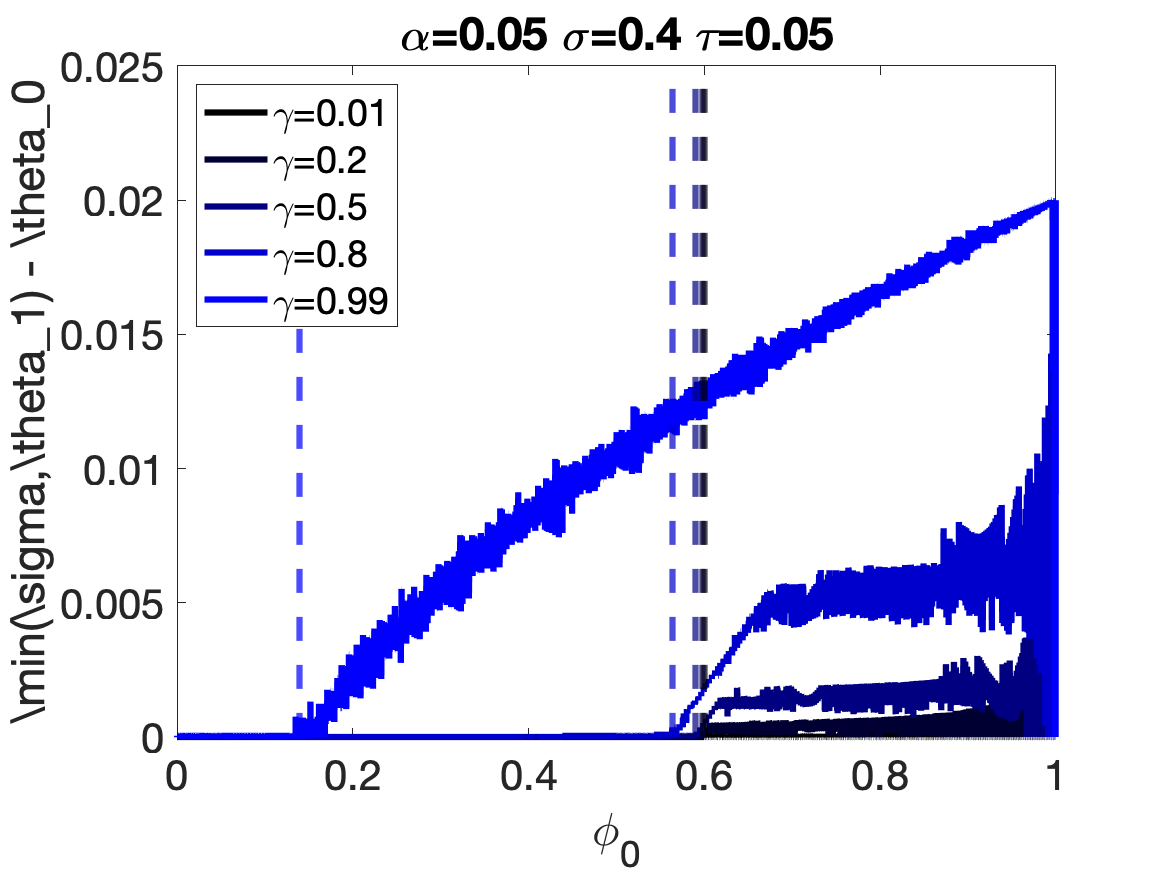}
    \end{subfigure}
    \begin{subfigure}[b]{0.31\textwidth}
        \includegraphics[width=\textwidth]{Figures/AA_Gamma_alpha_15_sigma_40_tau_10.png}
    \end{subfigure}
    \begin{subfigure}[b]{0.31\textwidth}
        \includegraphics[width=\textwidth]{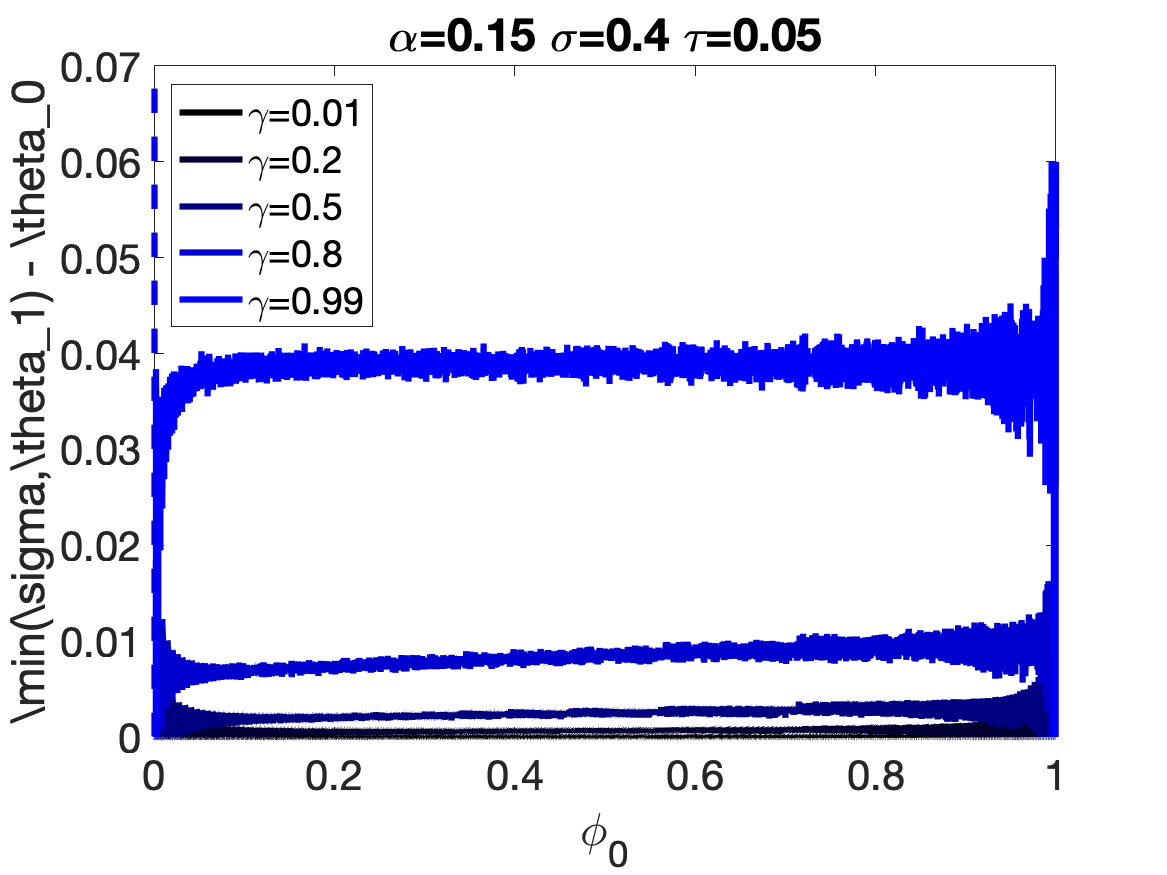}
    \end{subfigure}
   \caption{The difference between $\theta_0$ and $\theta_1$ at every state $0 \leq \phi_0 \leq 1$  for various settings of $\alpha, \sigma, \tau$, and $\gamma$. The dashed lines specify the tipping point, $\phi^*_0$. Note that strict affirmative action is only employed beyond $\phi^*_0$. Also note that \textbf{the extent of affirmative action is monotonically increasing in $\phi_0$}.}\label{fig:AA}
\end{figure*}

\begin{figure*}[h!]
    \centering
    \begin{subfigure}[b]{0.31\textwidth}
        \includegraphics[width=\textwidth]{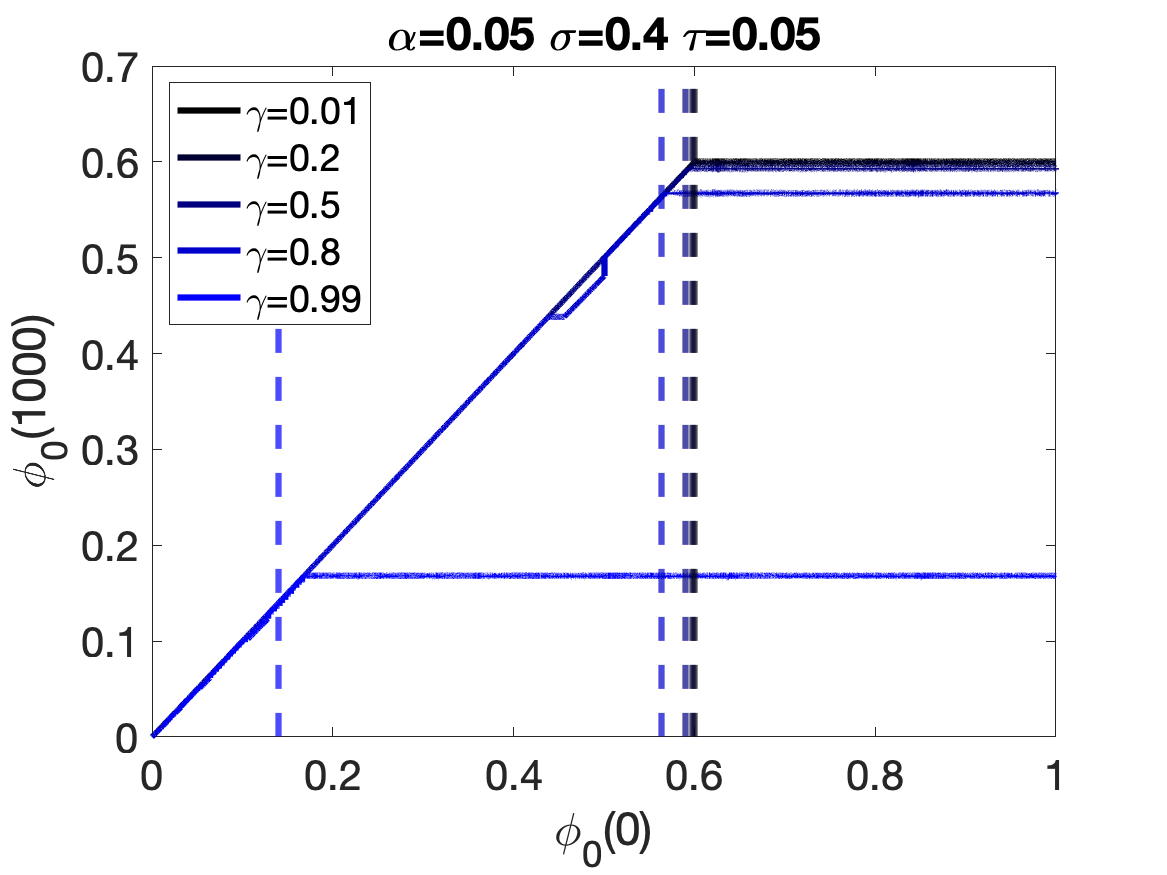}
    \end{subfigure}
    \begin{subfigure}[b]{0.31\textwidth}
        \includegraphics[width=\textwidth]{Figures/Convergence_Gamma_alpha_15_sigma_40_tau_10.png}
    \end{subfigure}
    \begin{subfigure}[b]{0.31\textwidth}
        \includegraphics[width=\textwidth]{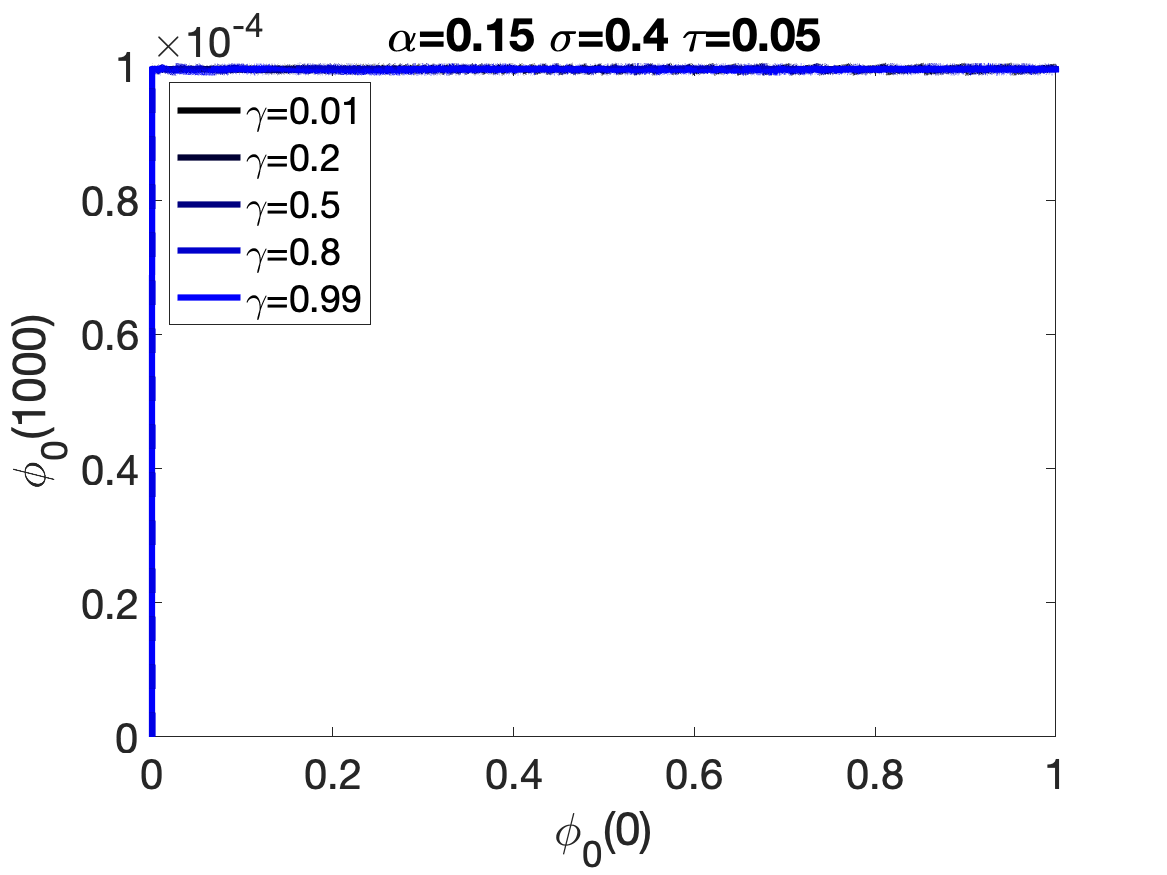}
    \end{subfigure}
   \caption{The state the optimal policy converges to with the initial state $\phi_0$. The dashed lines specify the tipping point, $\phi^*_0$. Note that \textbf{the optimal policy never shrinks the size of group $D$ to a value less than $\phi^*_0$}.}\label{fig:convergence}
\end{figure*}

\subsection{Omitted Proofs}\label{app:omitted_proofs}

\xhdr{Proof of Lemma~\ref{lem:phi_0_star}}
\begin{proof}
We utilize \emph{backward induction} to pinpoint the largest state at which $\sigma$ is an optimal threshold.
First, observe that if $\phi_0$ if sufficiently small, $\sigma \in \Pi^*_0(\phi_0)$. To see this, note that at least for $\phi_0=0$, $\sigma \in \Pi^*_0(\phi_0)$. Second, note that if $\sigma \in \Pi^*_0(\phi_0)$), Bellman optimality (\ref{eq:bellman}) must hold:
\begin{eqnarray*}
V(\phi_0) &=&  R(\phi_0, \sigma) + \gamma V(S(\phi_0, \sigma)) \\
&=&  R(\phi_0, \sigma) + \gamma V(\phi_0),
\end{eqnarray*}
where in the second line we utilized the fact that $S(\phi_0, \sigma) = \phi_0$ for all $\phi_0 \in [0,1]$. This fact can be readily verified through (\ref{eq:state}).
Rearranging the above equation, we obtain:
$$V(\phi_0) = \frac{1}{1-\gamma}R(\phi_0, \sigma).$$
Replacing $R$ above with its definition (\ref{eq:reward}) and plugging $\sigma$ in place of $\theta_0$, we obtain:
{\scriptsize
\begin{eqnarray}\label{eq:value}
V(\phi_0) &=& \frac{1}{1-\gamma}\frac{1-\phi_0}{2\sigma} \left((\sigma+\tau)^2 - \left( \frac{\sigma(1-\alpha) + (1-\phi_0)\tau - \phi_0 \sigma}{1-\phi_0}\right)^2\right) \nonumber\\
&=& \frac{1}{1-\gamma}\frac{1-\phi_0}{2\sigma} \left((\sigma+\tau)^2 - \left( \frac{-\sigma\alpha + (1-\phi_0)\tau  + (1- \phi_0) \sigma}{1-\phi_0}\right)^2\right)\nonumber\\
&=& \frac{1}{1-\gamma}\frac{1-\phi_0}{2\sigma} \left((\sigma+\tau)^2 - \left(  \sigma + \tau - \frac{\sigma\alpha}{1-\phi_0}\right)^2\right)\nonumber\\
&=& \frac{1}{1-\gamma}\frac{1-\phi_0}{2\sigma} \frac{\sigma\alpha}{1-\phi_0} \left(2(\sigma+\tau) - \frac{\sigma\alpha}{1-\phi_0}\right)\nonumber\\
&=& \frac{\alpha}{(1-\gamma)} \left((\sigma+\tau) - \frac{\sigma\alpha}{2(1-\phi_0)}\right)
\end{eqnarray}
}
A corollary of the above is that
$$V(0) = \frac{\alpha}{(1-\gamma)} \left((\sigma+\tau) - \frac{\sigma\alpha}{2}\right).$$

Next, we derive the largest $\phi_0$ at which $\sigma$ remains an optimal threshold. Let's denote this point by $\tilde{\phi}_0$.
From Bellman optimality, we know that an action is optimal if and only if it is optimal with respect to the (optimal) value function. So taking (\ref{eq:value}) as the value function $V$ up to $\tilde{\phi}_0$, any optimal threshold at $\phi_0 \leq \tilde{\phi}_0$ must maximize $R(\phi_0, \theta_0) + \gamma V(S(\phi_0, \theta_0))$. Since $R(\phi_0, \theta_0) + \gamma V(S(\phi_0, \theta_0))$ is differentiable and $\tilde{\phi}_0$ is the largest state at which $\sigma$ is an optimal policy, $\sigma$ must satisfy the first order condition at $\tilde{\phi}_0$---that is, the partial derivative of $R(\phi_0, \theta_0) + \gamma V(S(\phi_0, \theta_0))$ (w.r.t. $\theta_0$) must be 0 at $(\tilde{\phi}_0, \sigma)$. 

We can derive the derivative of $R(\phi_0, \theta_0) + \gamma V(S(\phi_0, \theta_0))$ as follows:
{\scriptsize
\begin{equation}\label{eq:partial}
\dd{\theta_0} \left\{ R(\phi_0, \theta_0) + \gamma V(S(\phi_0, \theta_0)) \right\} =  \dd{\theta_0}  R(\phi_0, \theta_0) + \gamma \dd{\theta_0}\  S(\phi_0, \theta_0) \dd{\phi_0} V(S(\phi_0, \theta_0)).
\end{equation}
}
We can calculate each term in (\ref{eq:partial}) as follows:
{\scriptsize
\begin{eqnarray}\label{eq:partial_reward}
\dd{\theta_0} R(\phi_0, \theta_0) &=& -\frac{\phi_0\theta_0}{\sigma} + 
\frac{\phi_0}{\sigma} \left( \frac{\sigma(1-\alpha) + (1-\phi_0)\tau - \phi_0 \theta_0}{1-\phi_0}\right) \nonumber\\
&=&  \frac{\phi_0}{\sigma} \left( \frac{\sigma(1-\alpha) + (1-\phi_0)\tau - \phi_0 \theta_0}{1-\phi_0} - \theta_0\right) \nonumber\\
&=&  \frac{\phi_0}{\sigma} \left( \frac{\sigma(1-\alpha) + (1-\phi_0)\tau - \theta_0}{1-\phi_0} \right).
\end{eqnarray}
}
\begin{equation}\label{eq:partial_state}
\dd{\theta_0} S(\phi_0, \theta_0) = \frac{\phi_0\theta_0}{\sigma}.
\end{equation}
\begin{equation}\label{eq:partial_value}
\dd{\phi_0} V(\phi_0) = -\frac{\sigma \alpha^2}{2(1-\gamma)(1-\phi_0)^2}.
\end{equation}
Plugging (\ref{eq:state}), (\ref{eq:partial_value}), (\ref{eq:partial_reward}), and (\ref{eq:partial_state}) into (\ref{eq:partial}), we obtain that:
{\scriptsize
\begin{eqnarray}\label{eq:derivative}
&& \dd{\theta_0} \left\{ R(\phi_0, \theta_0) + \gamma V(S(\phi_0, \theta_0)) \right\} = \nonumber\\
&& \frac{\phi_0}{\sigma} \left( \frac{\sigma(1-\alpha) + (1-\phi_0)\tau - \theta_0}{1-\phi_0} \right) 
- \gamma \frac{\phi_0\theta_0}{\sigma}
\frac{\sigma \alpha^2}{2(1-\gamma)(1-\phi_0 + \frac{\phi_0}{2\sigma} (\sigma^2 - \theta_0^2))^2}.
\end{eqnarray}
}

As mentioned earlier, at $(\tilde{\phi}_0, \sigma)$, the derivative (\ref{eq:derivative}) must amount to 0. Therefore, to find $\tilde{\phi}_0$, we must solve the following equation (obtained by replacing $\theta_0$ in (\ref{eq:derivative}) with $\sigma$):
{\scriptsize
 \begin{eqnarray*}
&& 0 = \frac{\phi_0}{\sigma} \left( \frac{\sigma(1-\alpha) + (1-\phi_0)\tau - \sigma}{1-\phi_0} \right) 
- \gamma \phi_0
\frac{\sigma \alpha^2}{2(1-\gamma)(1-\phi_0)^2}\\
 &\Rightarrow & 0 =  (-\alpha\sigma + (1-\phi_0)\tau)(1-\phi_0) - \gamma\frac{\sigma^2 \alpha^2}{2(1-\gamma)}\\
&\Rightarrow & 0 = \tau (1-\phi_0)^2 - \alpha\sigma (1-\phi_0)
- \frac{\gamma\sigma^2 \alpha^2}{2(1-\gamma)} \\
%
%
\end{eqnarray*}
}
(Note that in the second line of the derivation above, we multiplied both sides by $\sigma(1-\phi_0)^2/\phi_0)$.)
Solving the above quadratic equation for $(1-\phi_0)$, we have:
\begin{equation*}
\tilde{\phi}_0 \in 1 - \frac{\alpha \sigma}{2\tau} \left( 1 \pm \sqrt{1 + \frac{2\tau\gamma}{1-\gamma}} \right)
\end{equation*}
To obtain the tightest bound, we pick the smaller value among the above two possibilities:
\begin{equation*}
\tilde{\phi}_0 = 1 - \frac{\alpha \sigma}{2\tau} \left( 1 + \sqrt{1 + \frac{2\tau\gamma}{1-\gamma}} \right)
\end{equation*}
Note that $\tilde{\phi}_0$ must always be between $0$ and $(1-\alpha)$ (due to the budget constraints illustrated in Figure~\ref{fig:actions}). So the above derivation only goes through if  $\tilde{\phi}_0 \leq (1-\alpha)$. Therefore, we have:
$
\phi^*_0 =  \max \left\{ 0, \min \left\{ 1-\alpha , \tilde{\phi}_0 \right\} \right\}.
$
\end{proof}

\xhdr{Proof of Proposition~\ref{prop:strict-AA}}
\begin{proof}
Note that according to Lemma~\ref{lem:AA-leq}, for all $\theta_0 \in \Pi^*_0(\phi_0)$ and all $\theta_1 \in \Pi^*_1(\phi_0)$, $\theta_0 \leq \theta_1$. It only remains to show that the inequality is strict. That is, for all $\theta_0 \in \Pi^*_0(\phi_0)$ and all $\theta_1 \in \Pi^*_1(\phi_0)$, $\theta_0 \neq \theta_1$. 

Suppose not and there exists $\theta_0 \in \Pi^*_0(\phi_0)$ and $\theta_1 \in \Pi^*_1(\phi_0)$, such that $\theta_0 = \theta_1$. According to Lemma~\ref{lem:AA-line}, this implies that $\theta_0 = \sigma(1-\alpha)+\tau(1-\phi_0)$. Next we show that $\theta_0 = \sigma(1-\alpha)+\tau(1-\phi_0)$ cannot be an optimal threshold at $\phi_0$.

Recall that 
$$\Pi^*_0(\phi_0) = \arg\max_{\theta_0} R(\phi_0, \theta_0) + \gamma V(S(\phi_0, \theta_0)).$$
So if $\sigma(1-\alpha)+\tau(1-\phi_0) \in \Pi^*_0(\phi_0)$, it must satisfy the following first-order condition:
{\scriptsize
\begin{eqnarray*}
\dd{\theta_0} \left( R(\phi_0, \theta_0) + \gamma V(S(\phi_0, \theta_0)) \right) =  \dd{\theta_0}  R(\phi_0, \theta_0) + \gamma \dd{\theta_0}\  S(\phi_0, \theta_0) \dd{\phi_0} V(S(\phi_0, \theta_0)) = 0
\end{eqnarray*}
}
But note that $\dd{\theta_0} R(\phi_0, \theta_0)$ at $\theta_0 = \sigma(1-\alpha)+\tau(1-\phi_0)$ is 0:
{\scriptsize
\begin{eqnarray*}
\dd{\theta_0} R(\phi_0, \theta_0) &=&  \frac{\phi_0}{\sigma} \left( \frac{\sigma(1-\alpha) + (1-\phi_0)\tau - \theta_0}{1-\phi_0} \right) \\
&=&  \frac{\phi_0}{\sigma} \left( \frac{\sigma(1-\alpha) + (1-\phi_0)\tau - \left(\sigma(1-\alpha)+\tau(1-\phi_0)\right)}{1-\phi_0} \right) \\
&=&  \frac{\phi_0}{\sigma} \left( \frac{0}{1-\phi_0} \right) = 0
\end{eqnarray*}
}
So $\theta_0 = \sigma(1-\alpha)+\tau(1-\phi_0)$ can only be optimal if 
$$\dd{\theta_0}\  S(\phi_0, \theta_0) \dd{\phi_0} V(S(\phi_0, \theta_0)) = 0.$$ But this equality cannot hold because
$\dd{\theta_0} S(\phi_0, \theta_0) = \frac{\phi_0\theta_0}{\sigma} > 0,$
and $\dd{\phi_0} V(S(\phi_0, \theta_0)) < 0$. So $\theta_0 = \sigma(1-\alpha)+\tau(1-\phi_0)$ cannot be the optimal threshold at $\phi_0$.
\end{proof}

\subsection{Derivation of $\gamma^*$, $\tau^*$, and $\alpha^*$}\label{app:star}
The precise derivation of $\phi_0^*$ in Equation \ref{eq:phi_star} allows us to gain insight into how the interaction between the primitive parameters of our model can promote or avert affirmative action. In Figure~\ref{fig:star}, we focus on the interactions among $\alpha, \tau, \gamma$ (for simplicity assuming that $\sigma = 1-\tau$) and illustrate the regimes of persistent affirmative action (i.e., $\phi_0^* \leq 0$). We define and investigate the following quantities:
\begin{itemize}
\item Given $\tau$ and $\alpha$, $\gamma^*$ specifies the minimum discount factor required for $\phi_0^* \leq 0$.
\item Given $\gamma$ and $\alpha$, $\tau^*$ specifies the maximum level of $\tau$ that can maintain $\phi_0^* \leq 0$.
\item Given $\tau$ and $\gamma$, $\alpha^*$ specifies the minimum level opportunities required for $\phi_0^* \leq 0$.
\end{itemize}
Next, we derive the above quantities using Equation \ref{eq:phi_star}. In what follows, we assume $\tau > \alpha \sigma$, because we have
\begin{eqnarray*}
1 - \frac{\alpha \sigma}{2\tau} \left( 1 + \sqrt{1 + \frac{2\tau\gamma}{1-\gamma}}\right)
& \leq &  1 - \frac{\alpha \sigma}{2\tau} \left( 1 + 1\right)\\
& \leq &  1 - \frac{\alpha \sigma}{\tau}
\end{eqnarray*}
and if $\tau < \alpha \sigma$, then $1 - \frac{\alpha \sigma}{\tau} < 1- 1 = 0$ and $\phi^*_0 = 0$.

\paragraph{Derivation of $\tau^*$}
Assuming that $\alpha > 0$, $\tau, \gamma \in (0,1)$, $\sigma = 1-\tau$ and $2\tau > \sigma \alpha$, 
{\scriptsize
\begin{eqnarray*}
&& 1 - \frac{\alpha \sigma}{2\tau} \left( 1 + \sqrt{1 + \frac{2\tau\gamma}{1-\gamma}}\right) \leq 0 \nonumber\\
&\Leftrightarrow & 1 \leq \frac{\alpha \sigma}{2\tau} \left( 1 + \sqrt{1 + \frac{2\tau\gamma}{1-\gamma}} \right) \nonumber\\
\text{(assuming $\alpha> 0$ and $0 < \tau < 1$)} &\Leftrightarrow & \frac{2\tau}{\alpha \sigma} \leq  1 + \sqrt{1 + \frac{2\tau\gamma}{1-\gamma}}  \nonumber\\
&\Leftrightarrow & \frac{2\tau}{\alpha \sigma} - 1\leq   \sqrt{1 + \frac{2\tau\gamma}{1-\gamma}}  \nonumber\\
\text{(assuming $2\tau > \alpha (1-\tau)$)}  &\Leftrightarrow & \left( \frac{2\tau}{\alpha \sigma} - 1 \right)^2 \leq   1 + \frac{2\tau\gamma}{1-\gamma} \nonumber\\
&\Leftrightarrow & \frac{4\tau^2}{\alpha^2 \sigma^2} - \frac{4\tau}{\alpha \sigma} + 1 \leq   1 + \frac{2\tau\gamma}{1-\gamma} \nonumber\\
&\Leftrightarrow & \frac{4\tau^2}{\alpha^2 \sigma^2} - \frac{4\tau}{\alpha \sigma} \leq    \frac{2\tau\gamma}{1-\gamma}  \nonumber\\
\text{(divide by $\tau$ assuming $\tau >0$)}   &\Leftrightarrow & \frac{2\tau}{\alpha^2 \sigma^2} - \frac{2}{\alpha \sigma} \leq    \frac{\gamma}{1-\gamma}  \nonumber\\
\end{eqnarray*}
\begin{eqnarray}
&\Leftrightarrow & 2\tau - 2\alpha \sigma \leq    \frac{\gamma \alpha^2 \sigma^2}{1-\gamma}  \nonumber\\
&\Leftrightarrow & 2\tau - 2\alpha (1-\tau) \leq    \frac{\gamma \alpha^2}{1-\gamma} (1-\tau)^2  \nonumber\\
&\Leftrightarrow & 2 - 2 + 2\tau - 2\alpha (1-\tau) \leq    \frac{\gamma \alpha^2}{1-\gamma} (1-\tau)^2  \nonumber\\
&\Leftrightarrow & 2 - 2(1-\tau) - 2\alpha (1-\tau) \leq    \frac{\gamma \alpha^2}{1-\gamma} (1-\tau)^2  \nonumber\\
&\Leftrightarrow & 0 \leq   \frac{\gamma \alpha^2}{1-\gamma} (1-\tau)^2  + 2(1+\alpha)(1-\tau) - 2 \nonumber\\
&\Leftrightarrow & 0 \leq   \frac{\gamma \alpha^2}{2(1-\gamma)} (1-\tau)^2  + (1+\alpha)(1-\tau) - 1 \nonumber\\
&\Leftrightarrow & \tau \leq   1- \frac{-(1+\alpha) + \sqrt{ (1+\alpha)^2 + \frac{2\gamma \alpha^2}{(1-\gamma)}}}{\frac{\gamma \alpha^2}{(1-\gamma)}} \label{eq:tau_star}
\end{eqnarray}
}
where the last line is derived by obtaining the roots of the quadratic function $q(x) = \frac{\gamma \alpha^2}{2(1-\gamma)} x^2  + (1+\alpha)x - 1$, as follows:
{\scriptsize
\begin{equation*}
x^*_{1} = \frac{-(1+\alpha) - \sqrt{ (1+\alpha)^2 + \frac{2\gamma \alpha^2}{(1-\gamma)}}}{\frac{\gamma \alpha^2}{(1-\gamma)}} \text{  ,  } x^*_{2} = \frac{-(1+\alpha) + \sqrt{ (1+\alpha)^2 + \frac{2\gamma \alpha^2}{(1-\gamma)}}}{\frac{\gamma \alpha^2}{(1-\gamma)}}.
\end{equation*}
}
Note that $\frac{\gamma \alpha^2}{2(1-\gamma)} > 0$, so $q(x) \geq 0$ if and only if $x < x^*_{1}$ or $x > x^*_{2}$. Note that $x^*_{1}<0$ so if we know $x$ is positive, then $q(x) \geq 0$ if and only if $x > x^*_{2}$. Replacing $x$ with $(1-\tau)$, gives us equation (\ref{eq:tau_star}).

\paragraph{Derivation of $\alpha^*$}
Assuming that $\tau \in (0,1)$ and $\gamma < 1$,
{\scriptsize
\begin{eqnarray*}
&& 1 - \frac{\alpha \sigma}{2\tau} \left( 1 + \sqrt{1 + \frac{2\tau\gamma}{1-\gamma}}\right) \leq 0\\
&\Leftrightarrow & 1 \leq \frac{\alpha \sigma}{2\tau} \left( 1 + \sqrt{1 + \frac{2\tau\gamma}{1-\gamma}} \right)\\
&\Leftrightarrow & \frac{2\tau}{\sigma} \left( 1 + \sqrt{1 + \frac{2\tau\gamma}{1-\gamma}} \right)^{-1} \leq  \alpha
\end{eqnarray*}
}

\paragraph{Derivation of $\gamma^*$}
Assuming that $\alpha, \tau, \sigma, \gamma \in (0,1)$ and $\tau > \sigma \alpha$, 
{\scriptsize
\begin{eqnarray*}
&& 1 - \frac{\alpha \sigma}{2\tau} \left( 1 + \sqrt{1 + \frac{2\tau\gamma}{1-\gamma}}\right) \leq 0\\
&\Leftrightarrow & 1 \leq \frac{\alpha \sigma}{2\tau} \left( 1 + \sqrt{1 + \frac{2\tau\gamma}{1-\gamma}} \right)\\
&\Leftrightarrow & \frac{2\tau}{\alpha \sigma} \leq  1 + \sqrt{1 + \frac{2\tau\gamma}{1-\gamma}} \\
&\Leftrightarrow & \frac{2\tau}{\alpha \sigma} - 1\leq   \sqrt{1 + \frac{2\tau\gamma}{1-\gamma}} \\
&\Leftrightarrow & \left( \frac{2\tau}{\alpha \sigma} - 1 \right)^2 \leq   1 + \frac{2\tau\gamma}{1-\gamma} \\
&\Leftrightarrow & \left( \frac{2\tau}{\alpha \sigma} - 1 \right)^2 - 1 \leq   \frac{2\tau\gamma}{1-\gamma} \\
&\Leftrightarrow & \frac{1}{2\tau}\left( \frac{2\tau}{\alpha \sigma} - 1 \right)^2 - \frac{1}{2\tau} \leq   \frac{\gamma}{1-\gamma} \\
&\Leftrightarrow & 1 - \left( \frac{1}{2\tau}\left( \frac{2\tau}{\alpha \sigma} - 1 \right)^2 - \frac{1}{2\tau} + 1 \right)^{-1} \leq   \gamma \\
\end{eqnarray*}
}

\end{document}